\documentclass[11pt]{article}
\usepackage[bookmarksopen=true,bookmarksnumbered=true,bookmarksopenlevel=4,CJKbookmarks=true]{hyperref}
\usepackage{multirow}
\usepackage{amssymb,amsmath,amsthm,mathrsfs,cite,cleveref,float,setspace,mathtools}
\usepackage{enumerate}
\usepackage[shortlabels]{enumitem}
\usepackage[top=0.8 in, left=0.6 in, right=1 in, bottom=0.8 in]{geometry}
\usepackage{lineno}
\usepackage[normalem]{ulem} % for strikeout
\usepackage[colorinlistoftodos]{todonotes}
\usepackage{mdframed}

\usepackage[noend]{algpseudocode}
\usepackage{algorithm,algorithmicx}
\makeatletter
\def\algbackskip{\hskip-\ALG@thistlm}
\makeatother

\definecolor{blue(pigment)}{rgb}{0.2, 0.2, 0.6}

\newcommand{\remove}[1]{\iftrue{{{\color{red} #1}}}\fi}

\newtheorem {theorem}{Theorem}
\newtheorem {lemma}{Lemma}
\newtheorem {proposition}{Proposition}
\newtheorem {corollary}{Corollary}

\usepackage{graphicx,caption,subcaption,epstopdf,enumitem,tikz,standalone}%stackengine
\usetikzlibrary{shapes,backgrounds}
\usetikzlibrary[positioning,decorations.pathmorphing]

\usepackage{stmaryrd} % for \arrownot

\newcommand{\layer}{\mathcal{L}}
\newcommand{\ecc}{e}
\newcommand{\FPC}{4-point condition}

\newcounter{mycomment}

\title{Eccentricity function in distance-hereditary graphs\footnote{Declarations of interest: none}}

\author{\medskip Feodor F. Dragan and Heather M. Guarnera   \\%} 
Algorithmic Research Laboratory, Department of Computer Science \\
Kent State University, Kent, Ohio, USA \\
\href{mailto:dragan@cs.kent.edu}
{\texttt{dragan@cs.kent.edu}},
\href{mailto:hmichaud@kent.edu}
{\texttt{hmichaud@kent.edu}}
}
\date{}
\begin{document}
\maketitle
%\linenumbers

\begin{abstract}
A graph $G=(V,E)$ is distance hereditary if every induced path of $G$ is a shortest path. 
In this paper, we show that the eccentricity function $\ecc(v)=\max\{d(v,u): u\in V\}$ in any distance-hereditary graph $G$  is almost unimodal, that is, every vertex $v$ with $\ecc(v)> rad(G)+1$ has a neighbor with smaller eccentricity. Here, $rad(G)=\min\{\ecc(v): v\in V\}$ is the radius of graph $G$.  
Moreover, we use this result to fully characterize the centers of distance-hereditary graphs.
Several bounds on the eccentricity of a vertex with respect to its distance to the center of $G$ or to the ends of a diametral path are established.
Finally, we propose a new linear time algorithm to compute all eccentricities in a distance-hereditary graph.

\medskip

\end{abstract}

\section{Introduction}
The eccentricity $\ecc(v)$ of a vertex~$v$ is the length of a longest shortest path from~$v$ to any other vertex.
In a distance-hereditary graph $G$, the length of any induced path between two vertices equals their distance in $G$~\cite{doi:10.1093/qmath/28.4.417}.
The diameter $diam(G)$ (maximum eccentricity) and radius $rad(G)$ (minimum eccentricity) of distance-hereditary graphs have been extensively studied. A close relationship between diameter and radius was discovered in~\cite{Dragan_1994,YuChepoi91}, where it was shown that $diam(G) \geq 2rad(G) - 2$.
It was shown in~\cite{Dragan_1994} %,YEH2003297} 
that with two sweeps of a Breadth-First Search (BFS) one can obtain a value that is very close to the diameter. In fact, any vertex~$v$ that is furthest from an arbitrary vertex~$u$ has eccentricity $\ecc(v) \geq diam(G) - 2$.
Later, Feodor Dragan and Falk Nicolai~\cite{DRAGAN2000191} showed that by using instead LexBFS (Lexicographic Breadth-First Search) one can get a vertex~$v$ (last visited by a LexBFS starting at any vertex $u$) with $\ecc(v) \geq diam(G) - 1$,
and additionally if $\ecc(v)$ is even, then $\ecc(v)$ exactly realizes the diameter of $G$.
This yielded a linear time algorithm to compute the diameter as well as a diametral pair of vertices~\cite{Dragan_1994,DRAGAN2000191}, i.e., a pair $x,y$ such that $d(x,y)=diam(G)$.
%of a arbitrary distance-hereditary graph. 
There is also a linear time algorithm to find a central vertex (a vertex with minimum eccentricity) and calculate the radius~\cite{Dragan_1994}.
These results were very recently generalized in~\cite{Coudert_2019}; it follows from~\cite{Coudert_2019} that all vertex eccentricities of a distance-hereditary graph $G$  can be computed in total linear time via a split decomposition of $G$. 

Here, we establish further properties of the eccentricity function in distance-hereditary graphs.
{Understanding the eccentricity function and being able to efficiently compute the diameter, radius, and all vertex eccentricities is of great importance.
For example, in the analysis of social networks
(e.g., citation networks or recommendation networks), biological systems (e.g., protein interaction networks),
computer networks (e.g., the Internet or peer-to-peer networks), transportation networks (e.g., public transportation
or road networks), etc., the {\it eccentricity} $e_G(v)$ of a vertex $v$ is used to measure the importance of $v$ in the network:
the {\em eccentricity centrality index} of $v$ \cite{Brandes} is defined as $\frac{1}{e_G(v)}$.}
%We define the locality $loc(v)$ of a vertex $v$ as the minimum length of a shortest path from~$v$ to a vertex with strictly smaller eccentricity.
A graph's eccentricity function is \emph{unimodal}
%if every non-central vertex has locality 1, that is, any
if every non-central vertex~$v$ has an adjacent vertex~$u$ with $\ecc(u) < \ecc(v)$.
The unimodality of the eccentricity function has been studied in a variety of graph classes;
for example, it is exactly unimodal in Helly graphs~\cite{FDraganPhD}
and almost unimodal in $(\alpha_1, \Delta)$-metric  graphs~\cite{Dragan2017EccentricityAT} (that includes all chordal graphs)  and in hyperbolic graphs~\cite{Alrasheed_2016}.
In particular, it was shown~\cite{Dragan2017EccentricityAT} that in a chordal graph~$G$ the unimodality of the eccentricity function can break for a (non-central) vertex~$v$, i.e., all neighbors $w$ of $v$ 
%can have~$loc(v) > 1$ only
satisfy $\ecc(w) \ge \ecc(v)$, only
under very specific conditions: that $diam(G) = 2rad(G)$, that $\ecc(v) = rad(G) + 1$, and that $v$ is at distance 2 from a central vertex. 
We show in the main theorem of Section~\ref{sec:unimodality} that the same conditions hold for vertices of distance-hereditary graphs. % with locality larger than 1.
This result, which is of independent interest, is a crucial intermediate step to establish the remaining results of this paper.

The center $C(G)$ (all vertices of $G$ with minimum eccentricity and the graph induced by those vertices) of a distance-hereditary graph $G$ is also of interest.
Many graph classes have a well defined center.
The center of a tree is either $K_1$ or $K_2$~\cite{Jordan_1869},
the center of a maximal outerplanar graph is one of seven special graphs~\cite{doi:10.1002/jgt.3190040108},
and more generally all possible centers of 2-trees are known~\cite{PROSKUROWSKI19801}.
Graph centers have also been characterized fully for chordal graphs~\cite{Chepoi_1988}.
%Yeh and Chang~\cite{YEH2003297} show that centers of distance-hereditary graphs are either cographs or have diameter at most 3, and also that any cograph is the center of some distance-hereditary graph.
In distance-hereditary graphs it is known~\cite{YuChepoi91} that the diameter of the center is no more than~3.
This was later improved by Hong-Gwa Yeh and Gerard Chang~\cite{YEH2003297} that either $diam(C(G))= 3$ and $C(G)$ is connected or $C(G)$ is a cograph (which may not be connected), i.e., a $P_4$-free graph. Furthermore, any cograph is the center of some distance-hereditary graph. 
We complete the characterization of centers of distance-herediatary graphs by investigating the instance that $C(G)$ is not a cograph (i.e., $diam(C(G)) = 3$),
in which case $C(G)$ takes the form of a graph~$H$ which is further described in Section~\ref{sec:centers},
and moreover each such~$H$ is the center of some distance-hereditary graph.
%\remove{
%\sout{
%We additionally describe how to find in linear time a central set $M \subseteq C(G)$
%which dominates $C(G)$ (every central vertex of $G$ is in $M$ or adjacent to a vertex in $M$).
%We also describe how one can find all central vertices $C(G)$ in linear time when %$diam(G)=2rad(G)$.
%}
%}
%
%\remove{
%$\sout{
%Finally, we provide two linear time algorithms to approximate all eccentricities.
%The first is an additive 2-approximation of the eccentricity of any vertex~$v$ based on the distances from $v$ to only two vertices $x,y$,  where $x$ and $y$ is an arbitrary pair of mutually distant vertices. 
%The second is an additive 1-approximation using the distance from~$v$ to a particular central subset $M \subseteq C(G)$. 
%}
%}
%\change{
Finally, we obtain several bounds on the eccentricity of an arbitrary vertex~$v$ with respect to its distance to a mutually distant pair and also its distance to the center $C(G)$.
A simple dynamic programming algorithm is also presented which computes all eccentricities in a distance-hereditary graph in total linear time by utilizing a pruning sequence which is a characteristic property of distance-hereditary graphs.
%}

{The main contributions of this paper are summarized as follows.
\begin{itemize}[noitemsep, nolistsep]
    % Section 3: Unimodality of the Eccentricity Function
    \item We show that the eccentricity function in distance-hereditary graphs is almost unimodal: for every vertex $v \notin C(G)$ there is a neighbor $w$ such that $e(w) < e(v)$ or $diam(G) = 2rad(G)$, $e(v)=rad(G)+1$, and $d(v,C(G))=2$. We present several consequences of this result for obtaining the eccentricity of a vertex in a distance-hereditary graph~$G$. \item We propose certificates for the diameter, the radius and all eccentricities and show that the eccentricity of any vertex is closely bounded by its distances to just two mutually distant vertices. 
    % Section 4: Certificates for eccentricities
    %For example, it is sufficient to know the set $C(G)$ to determine the diameter of~$G$, the set of all vertices $D(G) = \{v \in V: e(v) = diam(G)\}$ is sufficient to determine the radius of~$G$, and the set of vertices with eccentricity at most $rad(G)+1$ is sufficient to determine all eccentricities in~$G$.
    % Section 5: Eccentricities, mutually distant pairs and distances to the center
    %\item For any mutually distant pair~$x,y$ and vertex~$v$, set $k := \max\{d(v,x), d(v,y)\}$. Then, the eccentricity of a vertex~$v$ satisfies $k \le e(v) \le k+2$. The upper bound can be improved in several situations: $e(v) \le k$ if $|d(v,x) - d(v,y)| \ge 2$ and $e(v) \le k + 1$ if $|d(v,x) - d(v,y)| = 1$ or $d(x,y)$ is odd.
    % Section 6: Centers of distance-hereditary graphs
    \item We fully characterize the centers of distance-hereditary graphs. The center of a distance-hereditary graph~$G$ is either a cograph or takes the form of a graph~$H$ described in Section~\ref{sec:centers}. Moreover, every cograph and each such~$H$ is the center of some distance-hereditary graph.
    % Section 7: Computing all eccentricities
    \item We present a new linear time algorithm to compute all eccentricities in a distance-hereditary graph. %We were not able to use distance properties of distance-hereditary graphs to obtain this linear time algorithm; we had to use a characteristic pruning sequence.
\end{itemize}
}

\section{Preliminaries}~\label{sec:preliminaries}
Let $G=(V,E)$ be an undirected, simple (without loops or parallel edges), connected graph. Let $n=|V|$ and $m=|E|$. 
A \emph{path} $P(v_0,v_k)$ is a sequence of vertices $v_0,...,v_k$ such that $v_iv_{i+1} \in E$ for all $i \in [0,k-1]$; its \emph{length} is $k$.
A graph~$G$ is connected if there is a path between every pair of vertices.
Let $d_G(x,y)$ be the distance between two vertices~$x$ and~$y$ in~$G$, that is, the length of a shortest path from~$x$ to~$y$.
A subgraph $H$ of a graph $G$ is called \emph{isometric} if the distance in $H$ between any of its two vertices equals their distance in $G$.  By $G- \{x\}$ we denote an induced subgraph of $G$ obtained from $G$ by removing a vertex $x\in V$.
The \emph{eccentricity} $\ecc_G(v)$ of a vertex~$v$ is the maximum distance from~$v$ to any vertex.
The subindex is omitted if~$G$ is known by context.
The \emph{diameter} ($diam(G)$) and \emph{radius} ($rad(G)$) of a graph~$G$ is the maximum and minimum eccentricity of a vertex, respectively.
The \emph{center} is the set of vertices whose eccentricities are minimum: $C(G) = \{v \in V : \ecc(v) = rad(G)\}$. It will be convenient to denote by $C(G)$ also the subgraph of $G$ induced by set $C(G)$. 
We define $C^k(G) = \{v \in V : \ecc(v) \leq rad(G) + k\}$.
We denote the set of furthest vertices from~$v$ as $F(v)=\{u \in V: d(u,v) = \ecc(v) \}$.
Vertices $x,y$ are considered to be \emph{mutually distant} if $x \in F(y)$ and $y \in F(x)$; they are called \emph{a pair of mutually distant vertices}. A pair $\{x,y\}$ is called \emph{a diametral pair} if $d(x,y)=diam(G)$. 
The \emph{interval} $I(x,y)= \{v\in V: d(x,v)+d(v,y)=d(x,y)\}$ is the set of all vertices that are on shortest paths between~$x$ and~$y$.
An interval \emph{slice} is defined as $S_k(x,y) = \{v \in I(x,y): d(x,v) = k\}$ for some non-negative integer~$k$.
We denote by $<S>$ the subgraph of~$G$ induced by the vertices~$S \subset V$. Let also $d(v,S)=\min\{d(v,u): u\in S\}$ and $diam(S)=\max\{d_G(x,y): x,y\in S\}$. 

The \emph{neighborhood} of~$v$ consists of all vertices adjacent to~$v$, denoted by~$N(v)$,
and the \emph{closed neighborhood} of~$v$ is defined as $N[v] = N(v) \cup \{v\}$.
The $k^{th}$ neighborhood of a vertex $v$ is the set of all vertices of distance~$k$ to~$v$, that is, $N^k(v)=\{u \in V: d(u,v)=k\}$.
%Whereas a \emph{disk} of radius~$k$ centered at vertex~$v$ is the set $D(v,k)=\{u \in V: d(u,v) \leq k\}$.
Whereas a \emph{disk} of radius~$k$ centered at a set~$S$ (or a  vertex) is the set of vertices of distance at most~$k$ to some vertex of~$S$, that is, $D(S,k) = \{u \in V: d(u,S) \leq k\}$.
A vertex~$v$ is said to be \emph{universal} to a set~$S$ if $N(v) \supseteq S$. Let $V=\{v_1,...,v_n\}$. 
For an $n$-tuple of non-negative integers $(r(v_1),...,r(v_n))$,
a subset $M\subseteq V$ is an {\emph r-dominating set} for a set $S\subseteq V$ in $G$ if and only if for every $v \in S$ there is a vertex~$u \in M$ with $d(u,v) \leq r(v)$. We also say that $M$ {\emph r-dominates} $S$ in $G$. 
If $r(v_i)=1$ for all $i$, then we say that $M$ {\emph dominates} $S$ in $G$.  
If $S=V$ then we say that $M$ $r$-dominates $G$.
Two vertex sets $A$ and $B$ of $G$ are said to be \emph{joined} if each vertex of $A$ is adjacent to every vertex of $B$. 
A vertex is \emph{pendant} if $|N(v)|=1$.
Two vertices~$v$ and~$u$ are \emph{twins} if they have the same neighborhood or the same closed neighborhood.
\emph{True twins} are adjacent; \emph{false twins} are not.
A graph is distance-hereditary if and only if each of its connected induced subgraphs is isometric~\cite{doi:10.1093/qmath/28.4.417}, that is,
the length of any induced path between two vertices equals their distance in~$G$.
These graphs are also characterized by a pruning sequence~\cite{Bandelt:1986:DG:8625.8626}: they can be dismantled by repeatedly removing either a pendant vertex or one vertex from a pair of twin vertices. % which eventually leads to a graph reduced to a single vertex.

The following propositions provide basic information on distance-hereditary graphs necessary for the next sections.

%%%%%%%%%%%%%%%%% STILL TO DEFINE??
% universal vertex?
% induced subgraph?
% cograph

\begin{proposition}\label{prop:dhgCharacterization}  \cite{Bandelt:1986:DG:8625.8626,doi:10.1137/0217032} 
For a graph $G$, the following conditions are equivalent:
\setlist{nolistsep}
\begin{enumerate}[noitemsep, label=(\roman*)]
	\item $G$ is distance-hereditary;
	\item\label{byForbiddenSubgraphs} The house, domino, gem, and the cycles $C_k$ of length $k \geq 5$ are not induced subgraphs of $G$ (see Figure  \ref{fig:dhg});
	\item\label{byDownNeighbors} For an arbitrary vertex $x$ of $G$ and every pair of vertices $v,u \in N^k(x)$, that are connected in the same component of the graph $< V \setminus N^{k-1}(x) >$,
	we have $N(v) \cap N^{k-1}(x) = N(u) \cap N^{k-1}(x)$;
	\item\label{byFourPointCondition} (\FPC) For any four vertices $u,v,w,x$ of $G$ at least two of the following distance sums are equal: $d(u,v) + d(w,x)$, $d(u,w) + d(v,x)$, and $d(u,x) + d(w,v)$.
	If the smaller sums are equal, then the largest one exceeds the smaller ones by at most 2.
	\item\label{byPruningSequence}
	% Let G be a finite graph with at least two vertices. Then G is distance-hereditary if and only G is obtained from K2 by a sequence of one- vertex extensions: attaching pendant vertices and twin vertices.
	$G$ can be reduced to one vertex graph by a pruning sequence of one-vertex deletions: removing a pendant vertex or a single vertex from a pair of twin vertices.
\end{enumerate}
\end{proposition}

\begin{figure}
    [htb] 
     \vspace*{-.2cm}
    \centering
    \includegraphics[]{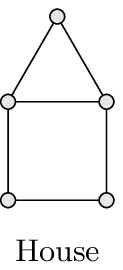}%
    \hspace*{1cm}%
    \includegraphics[]{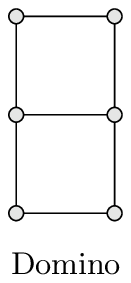}%
    \hspace*{1cm}%
    \includegraphics[]{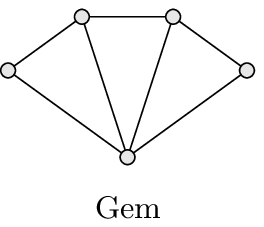}%
    \caption
    {
        Forbidden induced subgraphs in a distance-hereditary graph.
    }
    \label{fig:dhg} %
\end{figure}

\begin{proposition}\label{prop:dominatingClique}\cite{Dragan_1994}
Let $G$ be a distance-hereditary graph with $n$-tuple $(r(v_1), ..., r(v_n))$ of non-negative integers and $M \subseteq V$.
If every vertex pair $u,v \in M$ satisfies $d(u,v) \leq r(u) + r(v) + 1$ then $M$ has an $r$-dominating clique $C$.
If every vertex pair $u,v \in M$ satisfies $d(u,v) \leq r(u) + r(v)$ then there exists either a single vertex or a pair of adjacent vertices which $r$-dominates $M$.
\end{proposition}

\begin{proposition}\label{prop:eccentricityOfFurthestWrtRadius}\cite{Dragan_1994}
{For every vertex~$v$ of a distance-hereditary graph~$G$, a furthest from~$v$ vertex~$u \in F(v)$ satisfies $e(u) \ge 2rad(G) - 3$.}
%For every vertex $u$ of a distance-hereditary graph $G$, a furthest from $u$ vertex $v \in F(u)$ satisfies $e(v) \geq 2rad(G) - 3$. 
\end{proposition}

\begin{proposition}\label{prop:slicesAreJoined}
Let $G$ be a distance-hereditary graph and $x,y \in V$.
Any vertex $v \in S_k(x,y)$ has $S_{k+1}(x,y) \subseteq N(v)$, i.e., neighboring interval slices are joined. 
\end{proposition}
\begin{proof}
Consider a vertex $u \in S_{k+1}(x,y) \cap N(v)$ and any other vertex $w \in S_{k+1}(x,y)$.
Then $u$ and $w$ are connected in $< V \setminus N^{k+1}(x) >$ via shortest paths $P(u,y)$ and $P(w,y)$. By Proposition~\ref{prop:dhgCharacterization}\ref{byDownNeighbors}, they share neighboring vertices in $S_k(x,y)$. Hence, $w \in N(v)$.
\end{proof}

Unless otherwise stated, the graph $G$ is assumed to be distance-hereditary in all subsequent results.

%%%%%%%%%%%%%%%%%%%%%%%%%%%%%%%%%%%%%%%%%%%%%%%%%%%%%%
%~~~~~%~~~~~~~%~~~~~%~~~~~~~%~~~~~%~~~~~~~%~~~~~%~~~~%
%~~~~~%~~~~~~~%~~~~~%~~~~~~~%~~~~~%~~~~~~~%~~~~~%~~~~%
%%%%%%%%%%%%%%%%%%%%%%%%%%%%%%%%%%%%%%%%%%%%%%%%%%%%%%
\section{Unimodality of the eccentricity function}\label{sec:unimodality}
Recall that the eccentricity function is unimodal in $G$ if every non-central vertex~$v$ of $G$ has a neighbor~$u$ such that $\ecc(u) < \ecc(v)$. 
%We defined the locality~$loc(v)$ of a vertex $v$ to be the minimum distance from~$v$ to a vertex with strictly smaller eccentricity or to a central vertex. 
Our main result of this section is that in distance-hereditary graphs any vertex with sufficiently large eccentricity does have a neighbor with strictly smaller eccentricity.
Unimodality can break only at vertices with eccentricity equal to $rad(G)+1$, but those vertices are close (within 2) to the center $C(G)$.
Moreover, this can only occur when  $diam(G)=2rad(G)$.
This property of the eccentricity function of distance-hereditary graphs aligns with known results for other graph classes, such as chordal graphs and the underlying graphs of 7-systolic complexes~\cite{Dragan2017EccentricityAT}.

Our proof will be based on the following two lemmas. 
%%%%%%%%%%%%%%%%%%%%%%%%%%%%%%%%%%%%%%%%%%%%%%%%%%%%%%
% e(x) = r+1  --->   d(x,C(G)) <= 2
%%%%%%%%%%%%%%%%%%%%%%%%%%%%%%%%%%%%%%%%%%%%%%%%%%%%%%
\begin{lemma}\label{lem:twoFromCenter}
%Let $G$ be a distance-hereditary graph.
If a vertex~$x$ of $G$ has $e(x) = rad(G) + 1$, then $d(x,C(G)) \leq 2$.
\end{lemma}
\begin{proof}
Let $c$ be a central vertex closest to~$x$.
Consider any vertex~$v \in S_1(c,x)$ and vertex~$u \in F(v)$ furthest from~$v$.
As~$v$ is not central, $d(v,u) = rad(G) + 1$ and, by distance requirements, $d(c,u) = \ecc(c) = rad(G)$.
Hence, $u \in F(c) \cap F(v)$.

First we claim that $d(c,x) \leq 3$.
Since $c \in I(v,u)$ and $v \in I(c,x)$, we have $d(u,v) + d(c,x) = d(u,c) + d(v,x) + 2$. 
Consider the \FPC{} on vertices $u,x,v,c$.
As two distance sums must be equal, then either $d(u,x) + d(c,v) = d(u,v) + d(c,x)$ or $d(u,x) + d(c,v) = d(u,c) + d(v,x)$.
We have $d(u,x) + d(c,v) \leq \ecc(x) + 1 = rad(G) + 2$. 
In the first case we get $d(u,x) + d(c,v) = d(u,v) + d(c,x) = d(u,c) + d(v,x) + 2 = rad(G) + d(v,x) + 2$. Hence, $d(v,x) \leq 0$ and, by the triangle inequality, $d(x,C(G)) \leq 1$.  
In the second case we get $d(u,x) + d(c,v) = d(u,c) + d(v,x) = rad(G) + d(v,x)$. Hence, $d(v,x) \leq 2$ and, by the triangle inequality, $d(x,C(G)) \leq 3$, establishing the claim.

Assume now that $d(c,x) = 3$ and consider $y \in S_1(v,x)$. 
We next claim that $\ecc(y) = rad(G) + 1$.
By the choice of~$c$, vertex~$y$ is non-central and so $e(y) \geq rad(G) + 1$.
Since $y \in N(v) \cap N(x)$ with $\ecc(v) = \ecc(x) = rad(G) + 1$,  by distance requirements, $\ecc(y) \leq rad(G) + 2$.
By way of contradiction assume that $\ecc(y) = rad(G) + 2$.
Consider a furthest vertex~$y^* \in F(y)$. 
By distance requirements, $d(v,y^*) = d(x,y^*) = rad(G) + 1$ and $d(c,y^*) = rad(G)$.
Since there is a $(v,x)$-path via vertex~$y$ in $<V \setminus N^{rad(G)}(y^*)>$, by Proposition~\ref{prop:dhgCharacterization}\ref{byDownNeighbors}, the neighbors of~$v$ and~$x$ in $N^{rad(G)}(y^*)$ are shared.
Therefore, $cx \in E$, contradicting with $d(c,x) = 3$. Thus,  $\ecc(y) = rad(G) + 1$ must hold.

We now obtain a general contradiction in two steps. Recall that $\ecc(y) = \ecc(v)=\ecc(x)=rad(G) + 1$ and $\ecc(c)=rad(G)$. 
First, consider the \FPC{} on vertices $y,y^*,c,v$.
Consider three sums: $d(y,y^*) + d(c,v) = rad(G) + 2$, $d(y^*,c) + d(y,v) \leq rad(G) + 1$, and $d(y^*,v) + d(c,y) \leq rad(G) + 3$.
Clearly, the first and the second sums are not equal.
If the second and the third sums are equal, then $d(y^*,v) = d(y^*,c) + d(y,v) - d(c,y) \leq rad(G) - 1$, contradicting with $d(y^*,y) = rad(G) + 1$.
Therefore, the first and the third sums are equal.
Then $d(y^*,v) = d(y,y^*) + d(c,v) - d(c,y) = rad(G)$. Let $P(y^*,v)$ be any shortest path between $y^*$ and $v$. Its length is $rad(G)$. Consider also the path $Q=P(y^*,v),y,x$ (extension of $P(y^*,v)$ that includes also $y$ and $x$). In distance-hereditary graphs, every induced path is a shortest path. As $d(x,y^*) \leq \ecc(x) = rad(G) + 1$, $Q$ cannot be induced.  As $d(y,y^*)=rad(G)+1$, vertex~$x$ must be adjacent to some vertex on $P(y^*,v)$. To avoid large induced cycles $C_k$ of length $k \geq 5$, $x$ must be adjacent to a vertex~$z \in P(y^*,v)$ which is a neighbor of $v$. 
Thus, $d(y^*,x) = rad(G)$.
Necessarily, $zc \notin E$  since $d(x,c)=3$.
We also have that $d(y^*,c) \leq \ecc(c) = rad(G)$.

Next, consider the \FPC{} on vertices $y^*,c,x,v$.
We have $d(y^*,c) + d(x,v) \leq rad(G) + 2$, $d(y^*,v) + d(c,x) = rad(G) + 3$, and $d(y^*,x) + d(c,v) = rad(G) + 1$.
Since at least two sums must be equal, necessarily $d(y^*,c) = d(y^*,x) + d(c,v) -  d(x,v) = rad(G) - 1$.
Now all distances from~$y^*$ are known. 
We have $d(c,y^*) = d(z,y^*) = rad(G) - 1$, $d(v,y^*) = d(x,y^*) = rad(G)$, and $d(y,y^*) = rad(G) + 1$.
% Feodor: this can be shortened %%%%%%%%%%%%
%Since there is a $(c,z)$-path in $<V \setminus %N^{rad(G)-2}(y^*)>$, then by %Proposition~\ref{prop:dhgCharacterization}\ref{byDownNeighbors} %the neighbors of~$c$ and~$z$ in $N^{rad(G)-2}(y^*)$ are shared.
%Therefore there is a vertex~$t \in N^{rad(G)-2}(y^*)$ adjacent %to~$c$ and~$z$.
%However the vertices~$\{y,v,x,z,c,t\}$ induce a domino, a %contradiction to %Proposition~\ref{prop:dhgCharacterization}\ref{byForbiddenSubgr%aphs}.
Since there is a $(v,x)$-path in $<V \setminus N^{rad(G)-1}(y^*)>$, by  Proposition~\ref{prop:dhgCharacterization}\ref{byDownNeighbors}, the neighbors of~$v$ and~$x$ in $N^{rad(G)-1}(y^*)$ are shared.
Therefore, vertices $x$ and $c$ must be adjacent, contradicting with $d(x,c)=3$. 

Obtained contradiction proves the lemma.  
\end{proof}

%%%%%%%%%%%%%%%%%%%%%%%%%%%%%%%%%%%%%%%%%%%%%%%%%%%%%%
% ECCENTRICITY DESCENDS UNTIL r+1, AND THEN IS WEIRD ONLY WHEN d = 2r
%%%%%%%%%%%%%%%%%%%%%%%%%%%%%%%%%%%%%%%%%%%%%%%%%%%%%%
\begin{lemma}\label{lem:locality}
%Let $G$ be a distance-hereditary graph.
If there is a non-central vertex~$v$ of $G$ such that each vertex $w \in N(v)$ has $\ecc(w) \ge \ecc(v)$,
%with $loc(v) > 1$,
then $\ecc(v) = rad(G) + 1$ and $diam(G) = 2rad(G)$.
\end{lemma}
\begin{proof}
Consider a vertex $u \in F(v)$, a vertex $x \in S_1(v,u)$ with minimal $|F(x)|$, and let $y \in F(x)$.
By assumption, $\ecc(x) \geq \ecc(v)$ and $v$ is non-central.
%By the assumption on the locality of $v$, $\ecc(x) \geq \ecc(v)$ and $v$ is non-central.
If $d(v,y) = d(x,y) + 1$, then $\ecc(v) \geq d(v,y) = d(x,y) + 1 = \ecc(x) + 1 \geq \ecc(v) + 1$, a contradiction.
Thus, $d(x,y) - 1 \leq d(v,y) \leq d(x,y)$.

Consider the \FPC{} on vertices $u,v,x,y$.
We have  $d(u,v) + d(x,y) = \ecc(v) + \ecc(x) \geq 2rad(G) + 2$, $d(u,y) + d(x,v) \leq 2rad(G) + 1$, and  $d(u,x) + d(v,y) = \ecc(v) - 1 + d(v,y) \leq \ecc(v) + \ecc(x) - 1$.
Since two sums must be equal, necessarily $d(u,y) + d(x,v) = d(u,x) + d(v,y)$.
Hence $2rad(G) + 1 \geq d(u,y) + d(x,v) = d(u,x) + d(v,y) \geq (\ecc(v) - 1) + (\ecc(x) - 1) \geq 2\ecc(v) - 2 \geq 2rad(G)$.
Therefore, $2rad(G) \geq d(u,y) \geq 2rad(G) - 1$ and $\ecc(v) \leq rad(G) + 3/2$.
Since eccentricity is an integer and~$v$ is non-central, $\ecc(v) = rad(G) + 1$ must hold.

It remains only to show that $diam(G) = 2rad(G)$.
If $d(u,y) = 2rad(G)$, we are done. So, assume that $d(u,y) = 2rad(G) - 1$. 
We get $d(y,v) = d(u,y) + d(x,v) - d(u,x) = rad(G)$, and so $\ecc(x) = rad(G) + 1$. Furthermore, $v\in I(x,y)$.
The length of path $Q=P(u,v) \cup P(v,y)$ (the concatenation of $P(u,v)$ with $P(v,y)$) is $2rad(G)+1$. As $d(u,y) = 2rad(G) - 1$, $Q$ is not an induced path. Hence, there are vertices $s \in P(v,u)$ and $w \in P(v,y)$ such that $sw \in E$. 
To avoid large induced cycles $C_k$ of length $k \geq 5$,  necessarily $s \in S_1(x,u)$ and $w \in S_1(v,y)$ must hold.
%% Feodor: w and s were mixed below, 
%% it needs to be w everywhere instead of s %%%%%%%%%%%%%%%%%%%%%%%%%%%%%%%%%%%%%
Then, $w$ belongs to $S_1(v,u)$ as well as $x$. Since  $y \in F(x) \setminus F(w)$ (note that $\ecc(w)\ge \ecc(v)=rad(G)+1$ by assumption),
%as $loc(v)>1$),
by minimality of $|F(x)|$, there is a vertex $t \in F(w) \setminus F(x)$. 
Hence, $d(t,x) < rad(G) + 1$ and $d(t,w) = \ecc(w) \geq rad(G) + 1$. 

Now consider the \FPC{} on vertices $x,y,w,t$.
We have $d(x,y) + d(w,t) = \ecc(x) + \ecc(w) \geq 2rad(G) + 2$, whereas $d(t,y) + d(w,x) \leq 2rad(G) + 2$ and $d(x,t) + d(w,y) \leq 2rad(G) - 1$.
As $d(x,y) + d(w,t) - d(x,t) - d(w,y) \geq 3$, then only the two largest sums are equal: $d(x,y) + d(w,t) = d(t,y) + d(w,x)$.
Hence, $diam(G) \geq d(t,y) = d(x,y) + d(w,t) - d(w,x) \ge  2rad(G)$. 
That is, $diam(G)=d(t,y)=2rad(G)$. \end{proof}

We are ready to prove the main result of this section. 

\begin{theorem}\label{thm:localityAndCloseToCenter}
%Let $G$ be a distance-hereditary graph.
%If there is a vertex~$v$ \rev{of $G$} with $loc(v) > 1$,
Every vertex~$v \notin C(G)$ either has an adjacent vertex~$w$ with $\ecc(w) < \ecc(v)$ or
%If there is a vertex~$v$ \rev{of $G$ where each vertex $w \in N(v)$ has $\ecc(w) \ge \ecc(v)$},
$\ecc(v) = rad(G) + 1$, $diam(G) = 2rad(G)$, and $d(v,C(G)) = 2$.
\end{theorem}
\begin{proof}
%If a vertex~$v$ has $loc(v) > 1$ then
If a non-central vertex~$v$ has no neighbors with smaller eccentricity then $d(v,C(G))> 1$ and, by   Lemma~\ref{lem:locality}, $\ecc(v) = rad(G) + 1$ and $diam(G) = 2rad(G)$. Thus, by Lemma \ref{lem:twoFromCenter}, $d(v,C(G))=2$. \end{proof}

%\begin{corollary}
%If a distance-hereditary graph~$G$ has $diam(G) < 2rad(G)$, then the eccentricity function is unimodal.
%\end{corollary}

%%%%%%%%%%%%%%%%%%%%%%%%%%%%%%%%%%%%%%%%%%%%%%%%%%%%%%
%~~~~~%~~~~~~~%~~~~~%~~~~~~~%~~~~~%~~~~~~~%~~~~~%~~~~%
%~~~~~%~~~~~~~%~~~~~%~~~~~~~%~~~~~%~~~~~~~%~~~~~%~~~~%
%%%%%%%%%%%%%%%%%%%%%%%%%%%%%%%%%%%%%%%%%%%%%%%%%%%%%%
\section{Certificates for eccentricities}\label{sec:certificates}
We obtain as a consequence of Theorem~\ref{thm:localityAndCloseToCenter} several new results for distance-hereditary graphs on
lower and upper certificates for eccentricities, which were introduced in~\cite{Dragan2018RevisitingRD}
as a way to compute exactly or approximately eccentricities in a graph by maintaining upper and lower bounds.
A set $L$ (set $U$) of vertices is a \emph{lower certificate} (respectively, an \emph{upper certificate}) for eccentricities of $G$ if it is used to obtain lower bounds (respectively, upper bounds) of eccentricities in $G$. 
Given all distances from a vertex~$v$ to all vertices in $L \cup U$ as well as the eccentricities of vertices in $U$, we have 
the following lower and upper bounds for the eccentricity of any vertex~$v$~\cite{Dragan2018RevisitingRD}:

\[
e_L(v) \leq e(v) \leq e^{U}(v),
\text{ where }
\begin{cases}
	\ecc^U(v) = \min_{x \in U} d(v,x) + \ecc(x), \\
	\ecc_L(v) = \max_{x \in L} d(v,x).
\end{cases}
\]

A lower certificate $L$ (an upper certificate $U$) is said to be \emph{tight} if  $\ecc_L(v) = \ecc(v)$ ($\ecc^U(v) = \ecc(v)$, respectively) for all $v \in V$. 
A \emph{diameter certificate} is a set~$U$ such that $\ecc^U(v) \leq diam(G)$ for all $v \in V$,
and therefore the diameter is realized by $\max_{v \in V}e^U(v)$.
A \emph{radius certificate} is a set~$L$ such that $\ecc_L(v) \geq rad(G)$ for all $v \in V$,
and therefore the radius is realized by $\min_{v \in V}e_L(v)$.
In what follows, we define the set of all diametral vertices of $G$ as $D(G)=\{v \in V: \ecc(v)=diam(G)\}$.

In this section we show that all eccentricities can exactly be  determined in distance-hereditary graphs by
computing distances from vertices of $C^1(G)$ to all vertices, since $C^1(G)$ forms a tight upper certificate.
We also show that in distance-hereditary graphs the set $C(G)$ is a diameter certificate
and the set $D(G)$ is a radius certificate (a kind of duality between $C(G)$ and $D(G)$).
This agrees with radius and diameter certificates in chordal graphs~\cite{Dragan2018RevisitingRD} 
but, as we show later, this does not hold for arbitrary graphs.

We use the following corollary to Theorem~\ref{thm:localityAndCloseToCenter}.
\begin{corollary}\label{cor:pathToFurthestIntersectsCenter}
Let $G$ be a distance-hereditary graph.
\setlist{nolistsep}
\begin{enumerate}[noitemsep, label=(\roman*)]
	\item If $diam(G) < 2rad(G)$ then, for every pair of vertices $v \in V$ and $u \in F(v)$, there is a vertex $w \in I(v,u) \cap C(G)$ such that $u \in F(w)$.
	\item %If $diam(G) \leq 2rad(G)$ then, [Feodor: changed]
	For every pair of vertices $v \in V\setminus C(G)$ and $u \in F(v)$,
	there is a vertex $w \in I(v,u) \cap C^1(G)$ such that $u \in F(w)$.
\end{enumerate}
\end{corollary}
\begin{proof}
Consider any vertex $v \in V$ and $u \in F(v)$ and proceed by induction on $k := \ecc(v)$.
If $k = rad(G)$, then $w = v$ and we are done.
If $k = rad(G) + 1$ and $diam(G) = 2rad(G)$ then again $w = v$ and we are done.
If $k > rad(G) + 1$ or $k = rad(G) + 1$ and $diam(G) < 2rad(G)$
then, by Theorem~\ref{thm:localityAndCloseToCenter}, a neighbor $z$ of $v$ with $\ecc(z) = k - 1$
satisfies $u \in F(z)$, and we can apply the induction hypothesis.
\end{proof}

\begin{lemma}
The set $C^1(G)$ is a tight upper certificate for all eccentricities of $G$.
\end{lemma}
\begin{proof}
The statement follows from Corollary~\ref{cor:pathToFurthestIntersectsCenter} and the definition of a tight upper certificate.
\end{proof}

\begin{lemma}\label{lem:dhgDiamCertificate}
The center $C(G)$ is a diameter certificate of $G$.
\end{lemma}
\begin{proof}
This is clear by Corollary~\ref{cor:pathToFurthestIntersectsCenter} if $diam(G) < 2rad(G)$.
Additionally, in any graph $G$ with $diam(G) = 2rad(G)$ all central vertices $c \in C(G)$ and every diametral pair of vertices $x,y$
satisfy $d(x,y) = d(x,c) + d(c,y) = d(x,c) + rad(G) = rad(G) + d(y,c) = 2rad(G)$.
\end{proof}

%We use the following general fact established from~\cite{Dragan2018RevisitingRD}.
%\begin{lemma}\label{lem:whenRadiusCertificate}~\cite{Dragan2018RevisitingRD}
%The set $D(G)$ is a radius certificate of any graph $G$
%if $diam(G) \geq 2rad(G) - 1$.\todo{or when $d \leq r+1$?}
%\end{lemma}

%It is known~\cite{Dragan2018RevisitingRD} that $D(G)$ is a radius certificate of any graph~$G$ if $diam(G) \geq 2rad(G) - 1$.
%We provide here a short proof for completeness.
%\begin{lemma}\label{lem:whenRadiusCertificate}
%The set $D(G)$ is a radius certificate of any graph $G$
%%if $diam(G) \leq rad(G) + 1$ or
%if $diam(G) \geq 2rad(G) - 1$.\todo{or when $d \leq r+1$?}
%\end{lemma}
%\begin{proof}
%By contradiction assume first that $diam(G) \geq 2rad(G) - 1$ and $D(G)$ is not a radius certificate.
%Then there is a vertex $u \in V$ such that $max_{v \in D(G)}d(u,v) < rad(G)$ and thus any $x,y$ diametral pair
%has $d(x,y) \leq d(x,u) + d(u,y) \leq (rad(G) - 1) + (rad(G) - 1) = 2rad(G) - 2$, a contradiction that $diam(G) \geq 2rad(G) - 1$.
%
%%If $diam(G) = rad(G)$ then $D(G) = G$, and we are done.
%%Consider that $diam(G) = rad(G) + 1$ and by contradiction assume $D(G)$ is not a radius certificate.
%%Then there is a vertex $u \in V$ such that $max_{v \in D(G)}d(u,v) < rad(G)$.
%%Necessarily $u \in C(G)$ since $e(u) < diam(G)$.
%%Let $d(u,s) = e(u)$; by choice of~$u$ necessarily its furthest vertex~$s$ has $e(s) = rad(G)$, otherwise $s \in D(G)$, a contradiction.
%\end{proof}

\begin{lemma}\label{lem:dhgRadiusCertificate}
The set $D(G)$ is a radius certificate of $G$.
\end{lemma}
\begin{proof}
We first show that $D(G)$ is a radius certificate for any graph $G$ if $diam(G) \geq 2rad(G) - 1$. 
%\todo{need to cite or is it ok to quickly reprove?}
If %$diam(G) \geq 2rad(G) - 1$ and 
$D(G)$ is not a radius certificate,
then there is a vertex $u \in V$ such that $max_{v \in D(G)}d(u,v) < rad(G)$.
Thus, for  any diametral pair $x,y$, $d(x,y) \leq d(x,u) + d(u,y) \leq (rad(G) - 1) + (rad(G) - 1) = 2rad(G) - 2$, a contradiction with $d(x,y)=diam(G) \geq 2rad(G) - 1$.

As in a distance-hereditary $G$, $diam(G) \geq 2rad(G) - 2$ holds~\cite{Dragan_1994,YuChepoi91}, it remains to consider only the case when $diam(G) = 2rad(G) - 2$.
Let $S$ be the set of vertices~$u$ such that $d(u,t) \leq rad(G) - 1$ for all $t \in D(G)$.
By contradiction assume $D(G)$ is not a radius certificate and therefore $S$ is not empty.
Let $u \in S$ be a vertex which minimizes $|F(u)|$.
Consider any diametral pair $x,y$ and furthest from $u$ vertex $v \in F(u)$.
Necessarily $v \notin D(G)$ by the choice of~$u$.
Since $d(x,y)=2rad(G)-2$, $d(u,x)\leq rad(G)-1$ and  $d(u,y)\leq rad(G)-1$, clearly $d(u,x) = d(u,y) = rad(G) - 1$.

Consider the 4-point condition on vertices $v,u,x,y$.
We have that the largest distance sum is $d(v,u) + d(x,y) = d(v,u) + 2rad(G) - 2 \geq 3rad(G) - 2$,
given that $d(v,x) + d(u,y) \leq d(x,y)-1 + rad(G) - 1 = 3rad(G) - 4$
and that $d(v,y) + d(u,x) \leq d(x,y)-1 + rad(G) - 1 = 3rad(G) - 4$.
Therefore, the smaller sums are equal, establishing $d(v,x) = d(v,y)$. 
Moreover, since the difference between the largest sum and the other sums is at most 2, we get $d(v,u)=rad(G)$ and $d(v,x)=d(v,y)=2rad(G)-3$. So, $u \in C(G)$. 
%%% Feodor: can be shortened (as 3 <-> 4)
%%%%%%%%%%%%%%%%%%%%%%%%%%%%%%%%%%%%%%%%%%%%%%
%If $d(v,u) = rad(G) + 1$, then $d(v,x) = d(x,y) = diam(G)$ and %therefore $v \in D(G)$, a contradiction.
%Hence $d(v,u) = rad(G)$ and 
%So, $u \in C(G)$. 
%Additionally $d(v,x) \geq d(v,u) + d(x,y) - d(u,y) - 2 \geq rad(G) + 2rad(G) - 2 - rad(G) + 1 - 2 = 2rad(G) - 3$.
%Again since the smaller sums differ by the largest by at most %two, then $d(v,x) \geq d(v,u) + d(x,y) - d(u,y) - 2 \geq %2rad(G) - 3$.
%
%If $d(v,x) = 2rad(G) - 2$ then $v \in D(G)$, a contradiction.
%Hence $d(v,x) = d(v,y) = 2rad(G) - 3$.

We claim that there is a vertex~$w$ such that $d(w,x) = rad(G) - 1$, $d(w,y) = rad(G) - 1$ and $v \notin F(w)$.
Fix arbitrary shortest path $P(x,u)$, $P(y,u)$ and $P(u,v)$. 
Since $d(v,x) < d(v,u) + d(u,x) = diam(G) + 1$, path $Q=P(x,u)\cup P(u,v)$ is not induced. Hence, there must exist a chord between shortest path $P(x,u)$ and shortest path $P(u,v)$.
Define vertices $t,w \in P(u,v)$, $s,z \in P(x,u)$, $q, p \in P(y,u)$, as shown in Figure ~\ref{fig:dhgRadiusCertificateProof}.
%To avoid large induced cycles $C_k$ of length $k \geq 5$ in~$G$, %%% Feodor: this is a wrong argument
%%%%%%%%%%%%%%%%%%%%%%%%%%%%%%%%%%%%%%%%%
Since $d(v,x)=2rad(G)-3$, 
we must have the chord $zt \in E$ or the chord $sw \in E$.
By the same argument, there must exist a chord between shortest path $P(y,u)$ and shortest path $P(u,v)$
which is realized by chord $pt \in E$ or $qw \in E$.
We note that if $zt \in E$ then $pt \notin E$ since $d(x,y) = 2rad(G)-2$. 
Up to symmetry, we have two cases as shown in Figure ~\ref{fig:dhgRadiusCertificateProof}.
In case (a) we have $zt,qw \in E$,
and in case (b) we have $sw,qw \in E$.
In either case vertex~$w$ satisfies the desired properties, establishing the claim.
\begin{figure}[!htbp] 
  \begin{center}
    \includegraphics[scale=0.7]{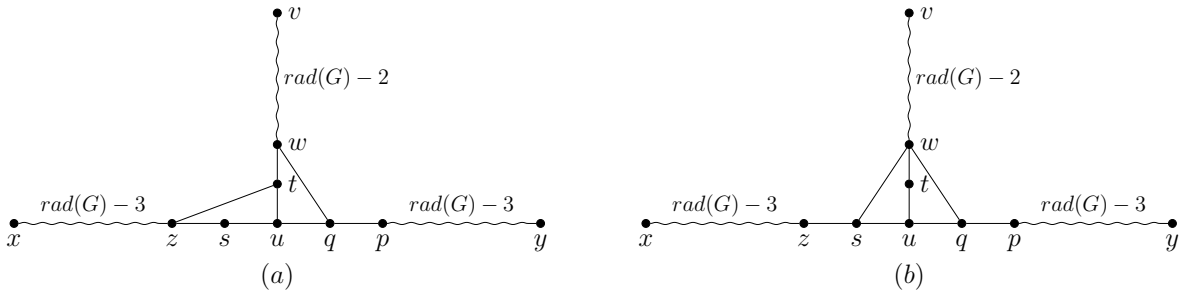}
\caption{Illustration to the proof of Lemma~\ref{lem:dhgRadiusCertificate}.}
    \label{fig:dhgRadiusCertificateProof}
  \end{center}
\end{figure}

We next claim that there is a vertex~$\overline{w}$ such that $d(w,\overline{w}) \geq rad(G)$ and $d(u,\overline{w}) \leq rad(G) - 1$.
On one hand, if $w \notin S$ then, by definition of $S$, there exists a vertex~$\overline{w} \in D(G)$ such that $d(w,\overline{w}) \geq rad(G)$ and, by the choice of~$u \in S$, we have $d(u,\overline{w}) \leq rad(G) - 1$.
On the other hand, if $w \in S$ then, by minimality of $|F(u)|$ and since $v \notin F(w)$, there exists a vertex $\overline{w} \in F(w) \setminus F(u)$. 
As $\overline{w} \notin F(u)$, $d(u,\overline{w}) \leq rad(G) - 1$ 
and, as $\overline{w} \in F(w)$,  $d(w,\overline{w}) \geq rad(G)$, establishing the claim. 

Consider now the \FPC{} on vertices $v,u,w,\overline{w}$.
Since $v \notin D(G)$ we have $d(v,\overline{w}) + d(w,u) \leq 2rad(G) - 3 + 2 = 2rad(G) - 1$.
We also have $d(v,w) + d(\overline{w},u) \leq rad(G) - 2 + rad(G) - 1 = 2rad(G) - 3$
and $d(v,u) + d(w,\overline{w}) \geq rad(G) + rad(G) = 2rad(G)$.
Given that $d(v,u) + d(w,\overline{w})$ is strictly larger than the other sums, it must differ from them by at most 2.
However, it differs by at least 3, giving a contradiction.
\end{proof}

%\remove{\sout{
%As consequences of these results, we have that if the set $C^1(G)$ of a distance-hereditary graph $G$ is known, then all vertex eccentricities in $G$ can be computed in total $O(|C^1(G)||E|)$ time by performing a BFS from each vertex of $C^1(G)$.  Similarly, if the set $C(G)$ ($D(G)$) is known, then the entire set $D(G)$ ($C(G)$, respectively) of $G$ can be computed in $O(|C(G)||E|)$ ($O(|D(G)||E|)$, respectively)  time.  Although linear time algorithms for computing a central vertex, the radius, a diametral pair  and the diameter of a distance-hereditary graph are known in literature[CITED], no algorithms are known that compute all  eccentricities (or even only the sets $C(G)$ or $D(G)$) in total linear time. In Section
% \ref{sec:approx}, we discuss how to efficiently approximate all eccentricities in a distance-hereditary graph.
%}}
%\change{
As a consequence of these results, if the set $C^1(G)$ of a distance-hereditary graph $G$ is known, then all vertex eccentricities in $G$ can straightforwardly be found by performing a BFS from each vertex of $C^1(G)$.  Similarly, if the set $C(G)$ ($D(G)$) is known, then the entire set $D(G)$ ($C(G)$, respectively) of $G$ can be found.
However, as we will discuss in Section~\ref{sec:allEccentricities}, there is a more efficient approach to compute all eccentricities of a distance-hereditary graph.
%}

We note that Lemma \ref{lem:dhgDiamCertificate} and Lemma \ref{lem:dhgRadiusCertificate} do not hold for general graphs, as illustrated in Figure~\ref{fig:badForRadiusCertificate} by a  graph $G$ with $diam(G)=6$ and $rad(G)=4$.
Here $D(G) = \{x,y\}$ and $C(G) = \{u\}$, and all other vertices have eccentricity 5.
However, $D(G)$ is not a radius certificate since $e_{D(G)}(u) = 3 < rad(G)$. 
Moreover, $C(G)$ is not a diameter certificate since $e^{C(G)}(v)=d(v,u) + \ecc(u) = 8 > diam(G)$.
\begin{figure}[H]
  \begin{center}
    \includegraphics[scale=0.7]{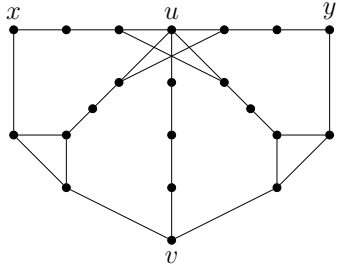}
\caption{A (non-distance-hereditary) graph $G$ where $D(G)$ is not a radius certificate and $C(G)$ is not a diameter certificate.
 %where (i) $D(G)$ is not a radius certificate, as demonstrated by the distance from vertices~$u$ and~$v$ to diametral vertices~$x$ and~$y$, and (ii) $C(G)$ is not a diameter certificate, as demonstrated by vertex~$v$.
 }
    \label{fig:badForRadiusCertificate}
  \end{center}
\end{figure}

One aims also to minimize the size of a certificate.
In %FEODOR: ????chordal graphs,?????
trees, for example, a single diametral pair is a sufficient radius certificate rather than the full set of diametral vertices.
Unfortunately this is not true %obtainable 
for distance-hereditary graphs.
The graph $G$ in Figure~\ref{fig:fullDiametralSetForRadiusCertificate} illustrates that every diametral vertex is necessary to establish a radius certificate. Graph $G$ consists of a clique of vertices $\{u_1,...,u_\ell\}$ and a clique of vertices $\{v_1,...,v_\ell\}$,
where each $u_i$ is adjacent to all vertices $v_{j \neq i}$, and each $u_i$ and $v_i$ has a pendant vertex $x_i$ and $y_i$, respectively.
$G$ is distance-hereditary as it can be dismantled via a sequence of pendant and twin vertex eliminations. %FEODOR
All vertices $x_i$ and $y_i$ are pendant,
each $u_i$ vertex is a false twin to~$v_i$,
and the remaining $v_i$ vertices are true twins (as they form a clique in the remaining graph).
Here $D(G)$ consists of all $x_i$ and $y_i$ vertices and $C(G)$ consists of all $u_i$ and $v_i$ vertices,
where $diam(G) = d(x_i,y_i) = 4$ and $rad(G) = d(v_i,x_i) = d(u_i,y_i) = 3$.
However, any $x_i \in D(G)$ has a vertex $v_i$ such that $d(v_i,t) < rad(G)$ for all $t \in D(G) \setminus \{x_i\}$.
By symmetry, the same is true for $y_i$ and its counterpart $u_i$.
Hence, all vertices of~$D(G)$ are necessary to form a radius certificate. 
%Feodor: New
One can also show that all vertices of $C(G)$ are necessary to form a diameter certificate. 
\begin{figure}[H]
  \begin{center}
    \includegraphics[scale=0.7]{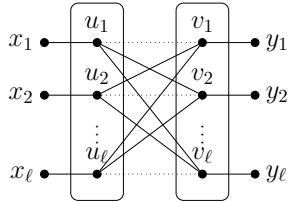}
\caption{A distance-hereditary graph $G$ for which $D(G) \setminus \{t\}$ is not a radius certificate for any $t \in D(G)$ and for which $C(G) \setminus \{c\}$ is not a diameter certificate for any $c \in C(G)$.}
    \label{fig:fullDiametralSetForRadiusCertificate}
  \end{center}
\end{figure}

%%%%%%%%%%%%%%%%%%%%%%%%%%%%%%%%%%%%%%%%%%%%%%%%%%%%%%
%~~~~~%~~~~~~~%~~~~~%~~~~~~~%~~~~~%~~~~~~~%~~~~~%~~~~%
%~~~~~%~~~~~~~%~~~~~%~~~~~~~%~~~~~%~~~~~~~%~~~~~%~~~~%
%%%%%%%%%%%%%%%%%%%%%%%%%%%%%%%%%%%%%%%%%%%%%%%%%%%%%%
\section{%Relation between mutually distant pairs, eccentricities, and distances to centers}
Eccentricities, mutually distant pairs and distances to the center}
\label{sec:eccFunction}
In this section, we show that the eccentricity of any vertex of a distance-hereditary graph is bounded by its distances to just two 
%vertices $x$ and $y$, where $\{x,y\}$ is an arbitrary  pair of 
mutually distant vertices. Furthermore, the distance between  any two mutually distant vertices is bounded by their distances to an arbitrary peripheral vertex (a vertex which is furthest for some other vertex).
The unimodality behavior of the eccentricity function described in Theorem~\ref{thm:localityAndCloseToCenter} gives also a relation between the eccentricity of a vertex
and its distance to $C^1(G)$.

%%%%%%%%%%%%%%%%%%%%%%%%%%%%%%%%%%%%%%%%%%%%%%%%%%%%%%
% Definining e(c) with respect to mutually distant pair
%%%%%%%%%%%%%%%%%%%%%%%%%%%%%%%%%%%%%%%%%%%%%%%%%%%%%%
\begin{lemma}\label{lem:duality1}
%Let $G$ be a distance-hereditary graph with a mutually distant pair $x,y$,
Let $x,y$ be a mutually distant pair of $G$,
and let $u \in V$ and $v \in F(u)$ be a furthest vertex from~$u$.
Let also $\alpha \coloneqq d(u,x)$ and $\beta \coloneqq d(u,y)$.
%Set $a \coloneqq d(c,x)$ and $b \coloneqq d(c,y)$.
Then,
\[\max\{\alpha,\beta\} \leq \ecc(u) \leq \max\{ \max\{\alpha,\beta\},\ \min\{\alpha,\beta\} + 2\} \leq \max\{\alpha,\beta\} + 2.\]
%\[\max\{d(c,x),d(c,y)\} \leq \ecc(c) \leq \max\{ \max\{d(c,x),d(c,y)\},\ \min\{d(c,x),d(c,y)\} + 2\} \leq \max\{d(c,x),d(c,y)\} + 2\]
Moreover, if $\ecc(u) = \max\{\alpha,\beta\} + 2$, then $\alpha=\beta=\ecc(u)-2$ and $d(v,x) = d(v,y) = d(x,y)$.
\end{lemma}
\begin{proof}
Let $v \in F(u)$. By the choice of~$v$, we have $\ecc(u) = d(u,v) \geq \max\{d(u,x), d(u,y)\}$.
Consider the 4-point condition on vertices $u,v,x,y$.
As $x$,$y$ is a mutually distant pair, we have for the three distance sums that $d(u,v) + d(x,y) = \ecc(u) + d(x,y)$,  $d(u,x) + d(v,y) \leq d(u,x) + d(x,y) \leq \ecc(u) + d(x,y)$,
and $d(u,y) + d(v,x) \leq d(u,y) + d(x,y) \leq \ecc(u) + d(x,y)$.
Clearly the first sum is largest.

% 1 = 2
We first consider the case when the first sum equals one of the latter.
Suppose that $d(u,v) + d(x,y) = d(u,x) + d(v,y)$.
Then, $\ecc(u) + d(x,y) = d(u,x) + d(v,y) \leq \ecc(u) + d(x,y)$.
Hence, $\ecc(u) = d(u,x) = \max\{d(u,x), d(u,y)\}$. % and $d(x,y) = d(y,v)$.
% 1 = 3
Suppose now that $d(u,v) + d(x,y) = d(u,y) + d(v,x)$.
Then, $\ecc(u) + d(x,y) = d(u,y) + d(v,x) \leq \ecc(u) + d(x,y)$.
Hence, $\ecc(u) = d(u,y) = \max\{d(u,x), d(u,y)\}$. % and $d(x,y) = d(x,v)$.
In either case, $\ecc(u) = \max\{d(u,x), d(u,y)\}$.

% 2 = 3
We next consider the case when the two smaller sums are equal and differ from the largest one by at most 2.
We have $\ecc(u) = d(u,v) \leq d(v,y) + d(u,x) - d(x,y) + 2=d(v,x) + d(u,y) - d(x,y) + 2$.
Since $d(x,y)$ is not smaller than $d(v,y)$ and $d(v,x)$, we obtain  $\ecc(u) \leq d(u,x) + 2$ and 
$\ecc(u) \leq d(u,y) + 2$, i.e., $\ecc(u) \leq \min\{d(u,x), d(u,y)\} + 2$.
Moreover, if $\ecc(u) = \max\{d(u,x), d(u,y)\} + 2$, we must be in the latter case when the two smaller sums are equal (otherwise, $\ecc(u) = \max\{d(u,x), d(u,y)\}$ as shown previously),
and so $d(u,v) = \ecc(u) = \min\{d(u,x), d(u,y)\} + 2$.
Hence, $d(u,x) = d(u,y) = d(u,v) - 2$ and, since $d(u,x) + d(v,y) = d(u,y) + d(x,v)$, $d(v,y) = d(x,v)$ holds too.
Combining this with the fact that $d(u,v) + d(x,y) - d(u,x) - d(y,v) \leq 2$, we obtain  $d(x,y) \leq d(v,y)$ 
and, since $x,y$ are mutually distant, necessarily $d(x,y)=d(v,y) = d(x,v)$, completing the proof.
\end{proof}

%%%%%%%%%%%%%%%%%%%%%%%%%%%%%%%%%%%%%%%%%%%%%%%%%%%%%%
% Definining d(x,y) with respect to distance to furthest vertex
%%%%%%%%%%%%%%%%%%%%%%%%%%%%%%%%%%%%%%%%%%%%%%%%%%%%%%
\begin{lemma}\label{lem:duality2}
Let $x,y$ be a mutually distant pair of $G$,
%Let $G$ be a distance-hereditary graph with a mutually distant pair $x,y$,
and let $u \in V$ and $v \in F(u)$ be a furthest vertex from~$u$. Let also $\alpha \coloneqq d(v,x)$ and $\beta \coloneqq d(v,y)$.
Then,
\[\max\{\alpha,\beta\} \leq d(x,y) \leq \max\{ \max\{\alpha,\beta\},\ \min\{\alpha,\beta\} + 2\} \leq \max\{\alpha,\beta\} + 2.\]
Moreover, if $d(x,y) = \max\{\alpha,\beta\} + 2$ then $\alpha=\beta=d(x,y)-2$ and $d(u,x) = d(u,y) = d(u,v)$.
\end{lemma}
\begin{proof}
%\todo{can we just state that it's the same as above as in first two sentences and then omit the rest?}
The proof is analogous to that of Lemma~\ref{lem:duality1} and is omitted. The only difference is that now we argue from the perspective of $d(x,y)$ and not of $\ecc(u)$. 
\end{proof}

%%%%%%%%%%%%%%%%%%%%%%%%%%%%%%%%%%%%%%%%%%%%%%%%%%%%%%
% Sharpness of bounds
%%%%%%%%%%%%%%%%%%%%%%%%%%%%%%%%%%%%%%%%%%%%%%%%%%%%%%
Figure~\ref{fig:dualityIsSharp}(a) illustrates that the upper bounds of Lemma~\ref{lem:duality1} and Lemma~\ref{lem:duality2} are sharp using vertices~$z$ and~$w$ for two opposing purposes.
First, $\ecc(z) = d(z,w) = \max\{d(x,z), d(y,z)\} + 2$, whereas $w \in F(z)$ has $d(x,y) = \max\{d(x,w), d(y,w)\}$.
Secondly, $\ecc(w) = d(z,w) = \max\{d(x,w), d(y,w)\}$, whereas $z \in F(w)$ has $d(x,y) = \max\{d(x,z), d(y,z)\} + 2$.
Recall the implications of Lemma~\ref{lem:duality1} and Lemma~\ref{lem:duality2} which state that
for a mutually distant pair $x,y$ and fixed vertices~$u \in V$ and~$v \in F(u)$,
if either $\ecc(u)$ or $d(x,y)$ is realized by its upper bound as given in the above inequalities, then the other value is realized by its lower bound.
So, it is not possible to obtain for the same $u \in V$ and~$v \in F(u)$ that both $\ecc(u) = \max\{d(x,u), d(y,u)\} + 2$ and $d(x,y) = \max\{d(x,v), d(y,v)\} + 2$ are true
 (Figure ~\ref{fig:dualityIsSharp}(a) uses a different starting vertex~$u$ to illustrate both upper bounds - with $u:=z$ and then with $u:=w$).
However, we show in Figure ~\ref{fig:dualityIsSharp}(b) an example when for fixed vertices~$u$ and~$v \in F(u)$, both $\ecc(u) = \max\{d(x,u), d(y,u)\} + 1$ and $d(x,y) = \max\{d(x,v), d(y,v)\} + 1$ are true.

\begin{figure}[H]
\centering
\includegraphics[scale=0.7]{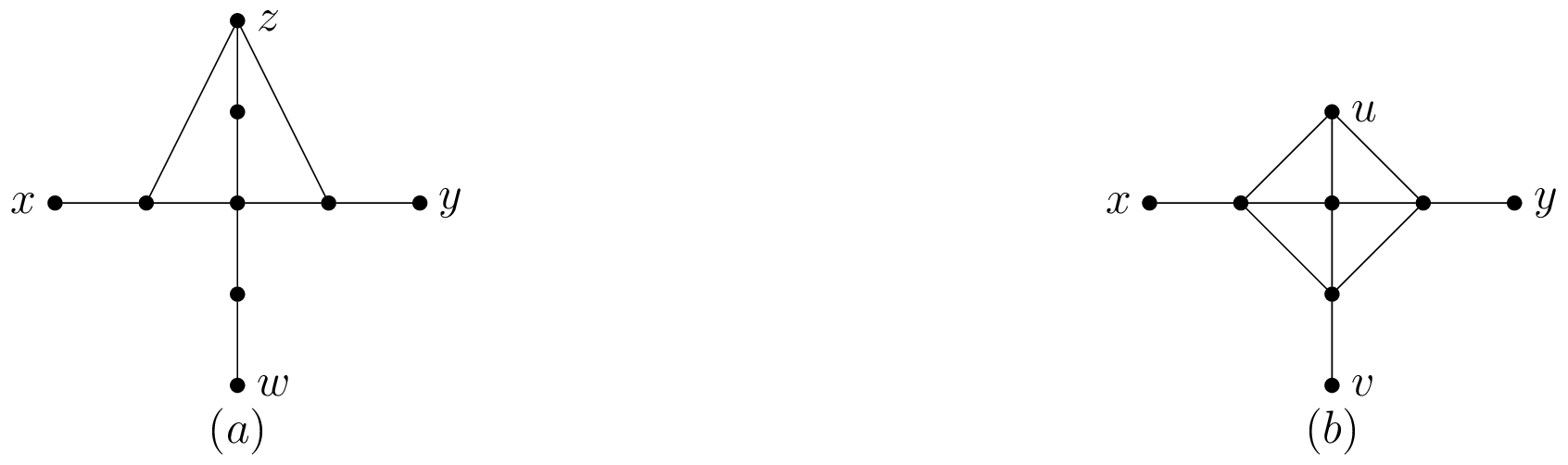}
\caption{Illustration to the sharpness of Lemma~\ref{lem:duality1} and Lemma~\ref{lem:duality2}.}
\label{fig:dualityIsSharp}
\end{figure}

In the case when~$x$ and~$y$ form a diametral pair, Lemma~\ref{lem:duality2} yields a result known from~\cite{YEH2003297}. We have $diam(G)=d(x,y)\le \max\{d(v,x),d(v,y)\}+2\le \ecc(v)+2$. 
\begin{corollary}~\cite{YEH2003297}\label{cor:vertexCloseToDiameter}
%Let $G$ be a distance-hereditary graph.
If vertex~$v$ of $G$ is a furthest vertex from any~$u \in V$, then $e(v) \geq diam(G) - 2$.
\end{corollary}

Corollary \ref{cor:vertexCloseToDiameter} can be used to find a pair of mutually distant vertices of a distance-hereditary graph in linear time. Pick an arbitrary start vertex $v_0$. With at most five BFSs find vertices $v_1,\dots, v_k$ ($2\le k\le 5$) such that $v_i\in F(v_{i-1})$ and $d(v_k,v_{k-1})=d(v_{k-1},v_{k-2})$. Since, by Corollary \ref{cor:vertexCloseToDiameter}, $\ecc(v_1)\ge diam(G)-2$, there are at most two improvements on $\ecc(v_1)$ to get a required mutually distant pair $v_{k-1},v_{k-2}$.

%\remove{
%\sout{
%Now, by Lemma \ref{lem:duality1},  the found distances from $v_{k-1}$ and $v_{k-2}$ to all vertices $u\in V$ can be used to get additive 2-approximations of all eccentricities in $G$ (see Section \ref{sec:approx}).
%}
%}
%\change{
Now, by Lemma \ref{lem:duality1}, the obtained distances from $v_{k-1}$ and $v_{k-2}$ to all vertices $u\in V$ already yield good lower bounds for vertex eccentricities in $G$ (they are within 2 from exact eccentricities).
We next show that, in some cases, they are even closer.

\begin{corollary}\label{cor:xycDiff}
%Let $x,y$ be a mutually distant pair of a distance-hereditary graph $G$ and let~$c \in V$.
Let $x,y$ be a mutually distant pair of $G$ and let~$u \in V$.
\setlist{nolistsep}
\begin{enumerate}[noitemsep, label=(\roman*)]
	\item If $|d(x,u) - d(y,u)| \geq 2$, then $\ecc(u) = \max\{d(x,u), d(y,u)\}$.
	\item If $|d(x,u) - d(y,u)| = 1$ or $u \in C(G)$, then $\max\{d(x,u), d(y,u)\} \leq e(u) \leq \max\{d(x,u), d(y,u)\} + 1$.
	%\item If $c \in C(G)$, then $\max\{d(x,c), d(y,c)\} \leq e(c) \leq \max\{d(x,c), d(y,c)\} + 1$.
\end{enumerate}
\end{corollary}
\begin{proof}
By Lemma~\ref{lem:duality1}, $\max\{d(x,u), d(y,u)\}\le \ecc(u) \leq \max\{\max\{d(x,u), d(y,u)\},\min\{d(x,u), d(y,u)\} + 2\}$. If $|d(x,u) - d(y,u)| \geq 2$,  then $\max\{d(x,u), d(y,u)\}\ge \min\{d(x,u), d(y,u)\} + 2$ and therefore $\ecc(u) = \max\{d(x,u), d(y,u)\}$. 
%
%If $e(c) \leq \max\{d(x,c), d(y,c)\}$ we are done.
%Otherwise by Lemma~\ref{lem:duality1}, $\ecc(c) \leq \min\{d(x,c), d(y,c)\} + 2$ and $\ecc(c) \geq \max\{d(x,c), d(y,c)\}$.
%If $|d(x,c) - d(y,c)| > 2$, then $\ecc(c) < \max\{d(x,c), d(y,c)\}$, a contradiction.
%If $|d(x,c) - d(y,c)| = 2$, then $\ecc(c) \leq \max\{d(x,c), d(y,c)\}$.
%
If $|d(x,u) - d(y,u)| = 1$, then $\ecc(u) \leq \max\{d(x,u), d(y,u)\} + 1$.
By contradiction assume now that $u \in C(G)$, $d(x,u) = d(y,u)$ and $e(u) = d(x,u) + 2$. 
By Proposition~\ref{prop:eccentricityOfFurthestWrtRadius}, we have $d(x,y) = e(x) \geq 2rad(G) - 3$.
By the triangle inequality, $d(x,y) \leq d(x,u) + d(u,y) = 2(rad(G) - 2) = 2rad(G) - 4$, a contradiction.
\end{proof}

%\begin{corollary}\label{cor:xyvDiff}
%Let $x,y$ be a mutually distant pair of a distance-hereditary graph $G$, let~$c \in V$ and $v \in F(c)$.\todo{unused - removed?}
%\setlist{nolistsep}
%\begin{enumerate}[noitemsep, label=(\roman*)]
%	\item If $|d(x,v) - d(y,v)| \geq 2$, then $d(x,y) = \max\{d(x,v), d(y,v)\}$. In particular, $\ecc(v) \geq d(x,y)$.
%	\item If $|d(x,v) - d(y,v)| = 1$, then $\max\{d(x,v), d(y,v)\} \leq e(c) \leq \max\{d(x,v), d(y,v)\} + 1$. In particular, $e(v) \geq d(x,y) - 1$.
%\end{enumerate}
%\end{corollary}
%\begin{proof}
%If $d(x,y) \leq \max\{d(x,v), d(y,v)\}$ we are done.
%Otherwise by Lemma~\ref{lem:duality2}, $d(x,y) \leq \min\{d(x,v), d(y,v)\} + 2$ and $d(x,y) \geq \max\{d(x,v), d(y,v)\}$.
%If $|d(x,v) - d(y,v)| > 2$, then $d(x,y)  < \max\{d(x,v), d(y,v)\}$, a contradiction.
%If $|d(x,v) - d(y,v)| = 2$, then $d(x,y) = \max\{d(x,v), d(y,v)\}$.
%If $|d(x,v) - d(y,v)| = 1$, then $d(x,y) \leq \max\{d(x,v), d(y,v)\} + 1$.
%In particular to either case, $e(v) \geq \max\{d(x,v), d(y,v)\}$.
%\end{proof}

In what follows, we analyze deeper the case when $d(u,x) = d(u,y)$. As the graph on Figure \ref{fig:dualityIsSharp} showed, in this case,  
$e(u) = \max\{d(x,u), d(y,u)\} + 2$ may happen. However, we demonstrate that it happens not very often. First we show that if $d(x,y)$ is odd, then still $e(u) \le \max\{d(x,u), d(y,u)\} + 1$. For this we will need one auxiliary lemma. 

\begin{lemma}\label{lem:sliceWhenEquidistantToOddXY}
%Let $G$ be a distance-hereditary graph, and let $c,x,y \in V$.
Let $u,x,y$ be vertices of~$G$.
If $d(u,x) = d(u,y)$ and $d(x,y) = 2k+1$ for some integer $k$,
then all vertices $s \in S_k(x,y)$ satisfy $d(u,x) = d(u,s) + k$.
\end{lemma}
\begin{proof}
Let $s \in S_k(x,y)$ and $d(u,x)=d(u,y)$.
necessarily,  $d(s,y) = k + 1$.
Consider the \FPC{} on vertices $x,y,u,s$.
We have  $d(u,s) + d(x,y) = d(u,s) + 2k + 1$,
 $d(u,x) + d(s,y) = d(u,x) + k + 1$,
and  $d(u,y) + d(x,s) = d(u,x) + k$.
Since at least two sums must be equal and the latter two sums are not, we consider the two remaining cases.
%Clearly $d(c,x) + d(s,y) \neq d(c,y) + d(x,s)$.
If $d(u,s) + d(x,y) = d(u,x) + d(s,y)$, then $d(u,x) = d(u,s) + (2k+1) - (k+1) = d(u,s) + k$, and we are done.
If $d(u,s) + d(x,y) = d(u,y) + d(x,s)$, then $d(u,x) = d(u,y) = d(u,s) + (2k+1) - k = d(u,s) + k + 1$.
However, by the triangle inequality, $d(u,x) \leq d(u,s) + d(s,x) = d(u,s) + k$, a contradiction.
\end{proof}

Next lemma handles the case when $d(x,y)$ is odd.

\begin{lemma}\label{lem:ecc1ApproxWhenOdd}
%Let $x,y$ be a mutually distant pair of a distance-hereditary graph $G$.
Let $x,y$ be a mutually distant pair of $G$.
If $d(x,y)$ is odd, then any vertex $u \in V$ has $e(u) \leq max\{d(x,u), d(y,u)\} + 1$.
\end{lemma}
\begin{proof}
%By Lemma~\ref{lem:duality1} we have that $e(c) \leq max\{d(x,c), d(y,c)\} + 2$.
Let $d(x,y) = 2k + 1$ for some integer $k$. By contradiction assume that $e(u) = \max\{d(x,u), d(y,u)\} + 2$. 
Consider a vertex $v \in F(u)$.
By Lemma~\ref{lem:duality1}, when $e(u) = \max\{d(x,u), d(y,u)\} + 2$, we have $d(u,x) = d(u,y) = e(u) - 2$, and $d(x,y) = d(x,v) = d(y,v)$.
Let $s \in S_k(x,y)$.
By Lemma~\ref{lem:sliceWhenEquidistantToOddXY} applied to vertex~$u$ and to vertex~$v$, we have $d(u,x) = d(u,s) + k$ and $d(v,x) = d(v,s) + k$.
By the triangle inequality, $e(u)=d(u,v) \leq d(u,s) + d(s,v) = (d(u,x) - k) + (d(v,x) - k) = d(u,x) + d(v,x) - 2k = d(u,x) + d(x,y) - 2k = d(u,x) +1$, contradicting with  $e(u) = d(u,x) + 2$.
\end{proof}

The following theorem summarizes the obtained bounds on the eccentricity of any vertex~$u \in V$. The distances from~$u$ to a selected pair of vertices is very close to the eccentricity of~$u$, and in some cases, measures it exactly.

\begin{theorem}\label{thm:allBounds}
Let $x,y$ be mutually distant vertices of $G$.
For every vertex~$u \in V$, the following holds:
\begin{equation*}
\max\{d(x,u), d(y,u)\} \le \ecc(u) \le \max\{d(x,u), d(y,u)\} + \begin{cases}
0, &\text{if $|d(x,u) - d(y,u)| \geq 2,$}\\
1, &\text{if $|d(x,u) - d(y,u)| = 1$ or $d(x,y)$ is odd,}\\
2, &\text{otherwise.}
\end{cases}
\end{equation*}
\end{theorem}

As a consequence of Theorem~\ref{thm:allBounds}, in distance-hereditary graphs, all eccentricities with an additive one-sided error of at most 2 can be computed in linear time. {We further investigated the case in which $d(x,u) = d(y,u)$ and $d(x,y)$ is even, but this case proved to be difficult and did not lead to any improvements on the upper bounds for $\ecc(u)$. % for $u \in V$.
}

We remark that a 2-approximation of eccentricities is known for distance-hereditary graphs. 
One common approach to approximating eccentricities in a graph $G$ is via an eccentricity $k$-approximating spanning tree~$T$ \cite{JGAA-Dragan,Dragan2018,Dragan2017EccentricityAT,Prisner2000}, i.e.,  a spanning tree $T$ of $G$ such that  $e_T(v) - e_G(v) \leq k$ holds for each vertex~$v$ of $G$.
 Note that every additive tree $k$-spanner (a spanning tree $T$ of $G$ such that $d_T(x,y)\le d_G(x,y)+k$ holds for every vertex pair $x,y$) is eccentricity $k$-approximating. 
 However, there are graph families which do not admit any additive tree $k$-spanners and yet they have very good eccentricity approximating spanning trees. The introduction of eccentricity approximating spanning trees is an attempt to weaken the restriction of additive tree spanners and instead closely approximate only distances to most distant vertices, the eccentricities. This is fruitful especially for those graphs for which additive tree $k$-spanners with small $k$ do not exist. For example, for every $k$ there is a chordal graph without an additive tree $k$-spanner, though every chordal graph has an eccentricity $2$-approximating spanning tree~\cite{Prisner2000,Dragan2017EccentricityAT} computable in linear time~\cite{Dragan2018}.
 More generally, all so-called $\delta$-hyperbolic graphs (note that chordal graphs are 1-hyperbolic) have an eccentricity $(4\delta+1)$-approximating spanning tree~\cite{JGAA-Dragan,hyperb-terrain}. As distance-hereditary graphs are also 1-hyperbolic, that general result for $\delta$-hyperbolic graphs implies that all distance-hereditary graphs have an eccentricity 5-approximating spanning tree. %However, 
 In fact, the situation with distance-hereditary graphs is  
 %different than with chordal graphs and 
 even simpler.  They have additive tree 2-spanners~\cite{Prisner_1997} (computable in linear time) and therefore eccentricity 2-approximating spanning trees. Furthermore, in general, the additive error 2 in an eccentricity approximating spanning tree cannot be improved.  
 %
 %, as the following example shows, such a spanning tree gives %already the best approximation of eccentricities by means of %spanning trees. %one may just as easily approximate all %distances.
 %An approach to compute a 2-spanner for distance-hereditary %graphs, AT-free graphs, and interval graphs was shown %in~\cite{Prisner_1997}.
 %However 
 A distance-hereditary graph in Figure~\ref{fig:dhgOneApproxTreeImpossible} has no eccentricity 1-approximating spanning tree. 
 Consider any edge $uv$ of the inner $C_4$ which is not present in $T$; either $e_T(u)=e_G(u)+2$ or $e_T(v)=e_G(v)+2$.
So, another approach is needed to efficiently compute all eccentricities in a distance-hereditary graph. In Section~\ref{sec:allEccentricities}, we present a new linear time algorithm for that which utilizes a characteristic pruning sequence. 
%computing all vertex eccentricities of a distance-hereditary graph.

 \begin{figure}[H]
   \begin{center}
     \includegraphics[scale=0.9]{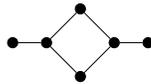}
     \caption{A distance-hereditary graph in which every spanning tree is eccentricity 2-approximating.}
     \label{fig:dhgOneApproxTreeImpossible}
   \end{center}
 \end{figure}

%%%%%%%%%%%%%%%%%%%%%%%%%%%%%%%%%%%%

We turn now to a relation between the eccentricity of a vertex
and its distance to $C(G)$ or $C^1(G)$. 

%%%%%%%%%%%%%%%%%%%%%%%%%%%%%%%%%%%%%%%%%%%%%%%%%%%%%%
% e(c) with respect to centers
%%%%%%%%%%%%%%%%%%%%%%%%%%%%%%%%%%%%%%%%%%%%%%%%%%%%%%
\begin{lemma}\label{lem:eccFunctionCentersEqn}
Let~$v \in V$ be any vertex of an arbitrary graph $G$ and let $k$ be an integer.
If $e(v) = d(v,C^k(G)) + rad(G) + k$, then $e(v) = d(v,C^{k+1}(G)) + rad(G) + k + 1$.
\end{lemma}
\begin{proof}
Suppose $e(v) = d(v,C^k(G)) + rad(G) + k$ and let $u \in F(v)$.
Let~$c \in  C^{k}(G)$ be a closest vertex to~$v$ in $C^{k}(G)$.
Consider an adjacent vertex $z \in S_1(c,v)$.
By the choice of~$c$, $e(z) = rad(G) + k + 1$ and therefore $z \in C^{k+1}(G)$.
Then,  $e(v) = d(v,c) + rad(G) + k = d(v,z) + d(z,c) + rad(G) + k \geq d(v,C^{k+1}(G)) + rad(G) + k + 1$.
By the triangle inequality,  also $e(c) \leq d(v,C^{k+1}(G)) + rad(G) + k + 1$.
\end{proof}

\begin{lemma}\label{lem:eccFunctionEqn}
Let $G$ be a distance-hereditary graph.
\setlist{nolistsep}
\begin{enumerate}[noitemsep, label=(\roman*)]
	\item If $diam(G) < 2rad(G)$, then all vertices~$v \in V$ satisfy $e(v) = d(v,C(G)) + rad(G)$.
	\item All vertices~$v \in V\setminus C(G)$ satisfy $e(v) = d(v,C^1(G)) + rad(G) + 1$.
	%%%%Feodor: changed 
\end{enumerate}
\end{lemma}
\begin{proof}
The statements follow from Theorem~\ref{thm:localityAndCloseToCenter} and Lemma~\ref{lem:eccFunctionCentersEqn} (see also Corollary \ref{cor:pathToFurthestIntersectsCenter}).
%%% At some point I had this separate proof, but I don't think it's needed anymore since I changed Lemma~\ref{lem:eccFunctionCentersEqn}
%%% to include <= 2rad(G) rather than strictly less than.
%The case when $diam(G) < 2rad(G)$ follows from Theorem~\ref{thm:localityAndCloseToCenter} and Lemma~\ref{lem:eccFunctionCentersEqn}.
%\todo{can't we say both cases follow from Thm and Lemma, and omit proof?}
%Let $diam(G) = 2rad(G)$.
%Consider a furthest vertex $u \in F(v)$.
%By Corollary~\ref{cor:pathToFurthestIntersectsCenter} there is a vertex~$w \in I(u,v) \cap C^1(G)$ with $d(u,w) = rad(G)+1$.
%By the triangle inequality, $d(v,C^1(G)) \leq d(v,w) = e(v) - rad(G) - 1$.
%Hence $e(v) \geq d(v,C^1(G)) + rad(G) + 1$.
%Again by the triangle inequality $e(v) \leq d(v,C^1(G)) + rad(G) + 1$.
\end{proof}

\begin{lemma}\label{lem:allPaths}
%Let $G$ be a distance-hereditary graph.
Let $v$ be an arbitrary vertex of $G$ and $v'$ be an arbitrary vertex of $C(G)$ closest to $v$.
Then,
$d(v,C(G)) + rad(G) - 1 \leq e(v) \leq d(v,C(G)) + rad(G)$. \newline
Furthermore, all shortest paths $P \coloneqq (v' = x_0, x_1, ..., x_\ell = v)$, connecting $v$ with $v'$, satisfy the following:
\setlist{nolistsep}
\begin{enumerate}[noitemsep, label=(\alph*)]
    \item if $e(v) = d(v, C(G)) + rad(G)$, then $e(x_i) = d(x_i,C(G)) + rad(G) = i + rad(G)$ for each $i \in \{0,...,\ell\}$;
    \item if $e(v) = d(v, C(G)) + rad(G) -1$, then $e(x_i) = d(x_i, C(G)) -1 + rad(G) = i - 1 + rad(G)$ for each $i \in \{3,...,\ell\}$ and $e(x_1) = e(x_2) = rad(G) + 1$.
\end{enumerate}
\end{lemma}
\begin{proof}
Let $r \coloneqq rad(G)$.
By the triangle inequality, $e(v) \leq d(v,C(G))+r$.
Let $e(v)\ge r+1$ and $w$ be a vertex of $C^1(G)$ closest to~$v$.
By the triangle inequality and Lemma~\ref{lem:twoFromCenter}, $d(v,v') = d(v,C(G)) \leq d(v,w) + d(w,C(G)) \leq d(v,C^1(G)) + 2$.
Combining with Lemma~\ref{lem:eccFunctionEqn}, one obtains $e(v) = d(v,C^1(G)) + r + 1 \geq d(v,C(G)) + r - 1$.

Consider an arbitrary shortest path $P \coloneqq (v' = x_0, x_1, ..., x_\ell = v)$ connecting $v$ with $v'$.
As adjacent vertices can have eccentricities which differ by at most one,
 $e(v) = d(v, C(G)) + r$ implies $e(x_i) = d(x_i,C(G)) + r = i + r$ for each $i \in \{0,...,\ell\}$.
If $e(v) = d(v, C(G)) + r -1$, there must exist an index $i \in \{1,...,\ell-1\}$ such that $e(x_i) = e(x_{i+1}) = r + i$
and $e(x_j) = r+j$ for $j \leq i$. We claim that $i=1$ must hold.

We first establish that $i \leq 2$.
By Theorem~\ref{thm:localityAndCloseToCenter}, there is a shortest path $Q^* \coloneqq (w = u_2, ..., u_\rho = v)$ connecting $v$ to
a vertex $w \in C^1(G)$ closest to $v$, where $d(w,C(G)) \leq 2$ and $e(u_k) = r + k - 1$ for all $k \in {2,...,\rho}$.
Let~$c$ be a vertex of $C(G)$ closest to~$w$.
We have that $d(v,C(G)) \leq d(v,c) \leq d(v,w) + d(w,c) \leq e(v) - e(w) + 2 = (d(v, C(G)) + r -1) - (r+1) + 2 = d(v,C(G))$.
Therefore, vertex~$c$ is also a closest to~$v$ central vertex, $d(w,C(G)) = 2$, and $\rho = \ell$.
Let now $Q \coloneqq (c = u_0, u_1, u_2, ..., u_\ell = v)$ be a shortest path which joins shortest paths $Q^*$ and any shortest path connecting $u_2$ to $c$.
As~$c$ is a closest to~$v$ central vertex, necessarily $e(u_1) = e(u_2) = r + 1$.
We claim that central vertices $u_0$ and $x_0$ are connected in $<V\setminus N^{d(v,C(G)) - 1}(v)>$ by respective paths to a vertex~$t \in F(v)$.
Let $u^* \in I(u_0,t)$, $u^* \neq u_0$. We have $d(u_0,t) \geq d(u^*,t) + 1$.
Assume that $d(u^*,v) \leq d(v,C(G)) - 1$.
Then, $e(v) = d(v,t) \leq d(v,u^*) + d(u^*,t) \leq (d(v,C(G)) - 1) + (d(u_0,t) - 1) \leq d(v,C(G)) + r - 2$, a contradiction with $e(v) = d(v,C(G)) + r - 1$.
Therefore, any vertex~$u^*$ on a shortest $(u_0,t)$-path satisfies $d(v,u^*) \geq d(v,C(G))$ and,
by symmetry, any vertex~$x^*$ on a shortest $(x_0,t)$-path also satisfies $d(v,x^*) \geq d(v,C(G))$.
The combined paths connect vertices~$u_0$ and~$x_0$ in $<V \setminus N^{d(v,C(G)) - 1}(v) >$.
As a result, for each $k \in \{0,...,\ell - 1\}$, vertices $u_k$ and $x_k$ are connected in $< V \setminus N^{d(v,C(G)) - k - 1}(v) >$.
Then, by Proposition~\ref{prop:dhgCharacterization}\ref{byDownNeighbors}, $u_{k}x_{k+1} \in E$ and $u_{k+1}x_{k} \in E$.
In particular, $u_2x_3 \in E$, $u_2x_1 \in E$, $u_1x_2 \in E$, and $x_0u_1 \in E$.
Recall that $e(u_2) = r + 1$.
If $i \geq 3$, then $e(x_3) = r + 3$; however, $e(x_3) \leq 1 + e(u_2) = r + 2$, a contradiction.

Assume now that $i=2$.
Hence, $e(x_2) = e(x_3) = r + 2$. Recall that $e(u_2) = e(u_1) = e(x_1) = r + 1$ and $e(x_0) = r$.
Let $z \in F(x_2)$; then $z$ is also furthest from vertices $u_1$, $x_1$, and $x_0$.
Denote by $Z$ a shortest path connecting $x_0$ and $z$, and let $z_0 \in Z$ be the vertex adjacent to $x_0$.
By distance requirements from each vertex $x_2$, $x_1$, and $u_1$ to furthest vertex~$z$, necessarily $z_0x_2 \notin E$, $z_0x_1 \notin E$, and $u_1z_0 \notin E$. 
Since $d(x_3,x_0) = 3$, $z_0x_3 \notin E$.
As $d(u_2,z) \leq e(u_2) = r + 1$, the path obtained by joining shortest paths $Z$ and $(x_0,u_1,u_2)$ by their common end-vertex $x_0$ is not induced.
The only possible chord is $u_2z_0$.
As $u_2$ and $x_2$ are connected in $<V \setminus N(x_0)>$, by Proposition~\ref{prop:dhgCharacterization}\ref{byDownNeighbors}, $z_0x_2 \in E$, a contradiction.
\end{proof}

%\remove{
%\sout{
%As a consequence of Lemma~\ref{lem:eccFunctionEqn}, one %may 
%can compute in linear total time  all eccentricities in a distance-hereditary graph $G$ exactly if all vertices of $C(G)$ and of $C^1(G)$ are known. Given $C^1(G)$, the distances $d(v,C^1(G))$ for all vertices $v\notin C(G)$ can simply be computed in total linear time by a BFS($C^1(G)$) starting at the set $C^1(G)$. Additionally, the radius $rad(G)$ of $G$ can be computed in linear time (see [CITE]). Thus, if $C(G)$ and $C^1(G)$ are known in advance, all eccentricities are computable in total linear time.   
%We use the following two sections then to describe the structure of centers of  distance-hereditary %graphs. The goal is to determine how a sufficiently large subset of $C^1(G)$ can be obtained that %will lead to additive 1-approximations of all eccentricities.  %
%We use the following two sections then to describe the structure of centers of  distance-hereditary graphs and how a sufficiently large subset of $C(G)$ 
%$C^1(G)$ 
%can be determined so that additive 1-approximations of all eccentricities can be computed.  %
%}
%}
%\change{
Lemma~\ref{lem:eccFunctionEqn} and Lemma~\ref{lem:allPaths} establish a close relationship between the eccentricity of a vertex and its distance to a closest vertex of $C(G)$ or of $C^1(G)$.
We now use the following section to fully describe the structure of centers of distance-hereditary graphs.
%}

%%%%%%%%%%%%%%%%%%%%%%%%%%%%%%%%%%%%%%%%%%%%%%%%%%%%%%
%~~~~~%~~~~~~~%~~~~~%~~~~~~~%~~~~~%~~~~~~~%~~~~~%~~~~%
%~~~~~%~~~~~~~%~~~~~%~~~~~~~%~~~~~%~~~~~~~%~~~~~%~~~~%
%%%%%%%%%%%%%%%%%%%%%%%%%%%%%%%%%%%%%%%%%%%%%%%%%%%%%%
\section{Centers of distance-hereditary graphs}\label{sec:centers}
In this section, we investigate the structure of centers  of distance-hereditary graphs and provide their full characterization. 
A subset $S \subseteq V$ is called \emph{$m^3$-convex} if and only if $S$ contains every induced path of length at least three between vertices of $S$.
%%% Feodor: some wrong things below...
%It is known from~\cite{Dragan1999ConvexityAH} that HHD-free graphs, including distance-hereditary graphs, satisfy that the disks $D(v,k)$, $k\geq 1$,
%are $m^3$-convex for all vertices~$v \in V$, and the set $D(S,k) = \cup_{v \in S}D(v,k)$, $k\geq 1$, are $m^3$-convex for all connected sets $S \subseteq V$.
%It is also known from~\cite{YEH2003297} that $C(G)$ is either a cograph or a connected graph with $diam(C(G)) \leq 3$.
%We remark that if the diameter of a set~$S$ in a distance-hereditary graph is no more than 2 then by definition it is a cograph.
%When $diam(C(G)) = 3$ then $C(G) = \cup_{c \in C(G)}D(c,3)$ and thus by $m^3$-convexity $C(G)$ is induced and therefore isometric.
It is known from~\cite{Dragan1999ConvexityAH} that the centers of HHD-free graphs are $m^3$-convex. As every distance-hereditary graph is HHD-free, the centers of distance-hereditary graphs are $m^3$-convex, too.
It is also known from~\cite{YEH2003297} that $C(G)$ is either a cograph or a connected graph with $diam(C(G)) = 3$.
 As every connected subgraph of a distance-hereditary graph is isometric, when $diam(C(G)) = 3$,  $C(G)$ is isometric and a distance-hereditary graph. We remark that if the diameter of a set~$S$ in a distance-hereditary graph is no more than 2 then, by definition, it induces a cograph (or, equivalently, a distance-hereditary graph of diameter at most 2).
%Any cograph is the center of a distance-hereditary graph~\cite{YEH2003297}.
%We expand this notion.

%In the remainder of this section we use~$x$ and~$y$ to denote an arbitrary pair of diametral vertices.

%\remove{
%\sout{
%To prove our main result of this section, we will need the following auxiliary lemmas, which will be used also in Subsection \ref{sec:1-appr}. 
%}
%}
%\change{
To prove our main result of this section, we will need the following auxiliary lemmas.
%}

%%%%%%%%%%%%%%%%%%%%%%%%%%%%%%%%%%%%%%%%%%%%%%%%%%%%%%%%%%%%%%%
% if d=2r then C(G) is cograph and belongs in middle slice,
%              C^1(G) belongs in slices...
%%%%%%%%%%%%%%%%%%%%%%%%%%%%%%%%%%%%%%%%%%%%%%%%%%%%%%%%%%%%%%%
\begin{lemma}\label{lem:centerWhen2r}
Let $x,y$ be a diametral pair of $G$, $diam(G) = 2rad(G)$,
%Let $G$ be a distance-hereditary graph with $diam(G) = 2rad(G)$, diametral pair $x,y$,
and $S \coloneqq S_{rad(G)}(x,y)$.
Then, $S$ and $C(G)$ are cographs with $C(G) \subseteq S$,
any vertex of $S_{rad(G)+1}(x,y) \cup S_{rad(G)-1}(x,y)$ is universal to $S$, 
and $C^1(G) \subseteq D(S,1)$.
\end{lemma}
\begin{proof}
Any central vertex~$c \in C(G)$ has $d(x,c) \leq rad(G)$ and $d(y,c) \leq rad(G)$, therefore, by distance requirements, $C(G) \subseteq S$.
By Proposition~\ref{prop:slicesAreJoined}, any vertex~$w \in S_{rad(G)+1}(x,y) \cup S_{rad(G)-1}(x,y)$ is universal to slice~$S$ and therefore universal to~$C(G)$.
%Thus $rad(C(G)) \leq 1$.
Moreover, since the diameters of $S$ and $C(G)$ are no more than 2, both are cographs.

Let now $c \in C^1(G)$ and by contradiction assume $c \notin D(S,1)$.
Necessarily $d(x,c) \leq rad(G) + 1$ and $d(y,c) \leq rad(G) + 1$.
If $d(x,c) < rad(G)$ then, by distance requirements, $d(y,c) = rad(G) + 1$ and $c \in S_{rad(G) - 1}(x,y)$. By Proposition~\ref{prop:slicesAreJoined}, $c \in D(S,1)$, a contradiction.
If $d(x,c) = d(y,c) = rad(G)$ then $c \in S$, a contradiction.
Hence, we can assume, without loss of generality, that $d(x,c) = rad(G) + 1$.
Consider vertex~$u \neq c$ on a shortest path $P(c,y)$ closest to~$x$, and let $b \in S_{rad(G)+1}(x,y)$.
Then, $d(y,u) \leq rad(G)$ as $d(y,c) \leq rad(G)+1$. 
If $d(x,u) \geq rad(G) + 1$, then vertices~$b$ and~$c$ are connected in $<V\setminus N^{rad(G)}(x)>$ and therefore, by Proposition~\ref{prop:dhgCharacterization}~\ref{byDownNeighbors}, share common neighbors in $S$. So, $c \in D(S,1)$, a contradiction. 
If $d(x,u) \leq rad(G)$, then $d(x,u) = d(y,u) = d(y,c) - 1 = rad(G)$ as $d(x,y) = 2rad(G)$. Hence, $u \in S \cap N(c)$, a contradiction.
\end{proof}

%%%%%%%%%%%%%%%%%%%%%%%%%%%%%%%%%%%%%%%%%%%%%%%%%%%%%%%%%%%%%%%
% if d=2r-1 then C(G) = A u B u C,
%                C(G) is cograph, 
%                D((a,b),1) includes C(G)
%                any c in C sees both (a,b)
%%%%%%%%%%%%%%%%%%%%%%%%%%%%%%%%%%%%%%%%%%%%%%%%%%%%%%%%%%%%%%%
\begin{lemma}\label{lem:centerWhen2r-1}
Let $x,y$ be a diametral pair of $G$, $diam(G) = 2rad(G) -1$,
%Let $G$ be a distance-hereditary graph with $diam(G) = 2rad(G) -1$ and a diametral pair $x,y$,
and let
$A \coloneqq S_{rad(G)-1}(x,y)$ and $B \coloneqq S_{rad(G)-1}(y,x)$.
Then, $C(G)$ is a cograph and any edge $ab \in E$, where~$a \in A$ and $b \in B$, satisfies $C(G) \subseteq D(\{a,b\},1)$.
Moreover, there is a vertex $a \in A \cap C(G)$ and a vertex $b \in B \cap C(G)$.
\end{lemma}
\begin{proof}
By Proposition~\ref{prop:slicesAreJoined}, slices~$A$ and~$B$ are joined.
Consider any $c \in C(G)$ and edge~$ab \in E$ for any vertex pair $a \in A$ and $b \in B$.
As $d(x,y)=2rad(G)-1$ and $\ecc(c)=rad(G)$, $d(x,c) < rad(G)$ implies $d(x,c) = rad(G) - 1$ and $d(y,c) = rad(G)$. Hence, $c \in A$.
By symmetry, if $d(y,c) < rad(G)$ then $c \in B$.
Assume now that $d(x,c) = d(y,c) = rad(G)$.
Then, $b,c \in N^{rad(G)}(x)$.
Consider vertex~$u \neq c$ on a shortest path $P(c,y)$ closest to~$x$.
We have  $d(y,u) \leq rad(G) - 1$ by the choice of~$u$.
%Consider for any~$u \neq c$ a shortest path $P(u,y)$; then  $d(y,u) \leq rad(G) - 1$.
If $d(x,u) \leq rad(G) - 1$, then $d(x,y) \leq d(x,u) + d(u,y) \leq 2rad(G) - 2$, a contradiction.
Thus, $d(x,u) \geq rad(G)$.
Vertices~$b$ and~$c$ are connected in $< V \setminus N^{rad(G)-1}(x) >$ by shortest paths to~$y$. By Proposition~\ref{prop:dhgCharacterization}~\ref{byDownNeighbors},
$b$ and~$c$ share neighbors in $A$. 
Therefore, $a \in N(c)$ and, by symmetry, $b \in N(c)$.
Hence, any central vertex either belongs to~$A \cup B$ or is universal to~$A \cup B$.
Thus, $C(G) \subseteq D(\{a,b\},1)$.
Additionally, since any pair of vertices in $C(G)$ is at most distance 2 apart, $C(G)$ is a cograph.

We now show the existence of vertices~$a \in A$ and $b \in B$ such that $a,b \in C(G)$.
Consider the family of disks $D(v,r(v))$ centered at each vertex~$v$,
where $r(v) = 1$ for all central vertices~$v \in C(G)$ and $r(v)=rad(G)-1$ for all others.
Any two non-central vertices~$u,v \in V \setminus C(G)$ have distance no more than the diameter, therefore $d(u,v) \leq 2rad(G) - 1 = r(u) + r(v) + 1$.
Any two central vertices~$u,v \in C(G)$ have distance no more than the diameter of the center, therefore $d(u,v) \leq 2 = r(u) + r(v)$.
By definition, any central vertex~$u \in C(G)$ sees any vertex~$v \in V$ within $rad(G)$, and therefore $d(u,v) \leq rad(G) = r(u) + r(v)$.
Hence, by Proposition~\ref{prop:dominatingClique}, there is an r-dominating clique~$K$.
As any non-central vertex has distance $rad(G)-1$ to a vertex of~$K$, we have $K \subseteq C(G)$.
Let~$a \in K$ be closest to~$x$ and let~$b \in K$ be closest to~$y$.
By distance requirements, $ab$ must be an edge with~$d(x,a)=rad(G)-1$ and $d(b,y)=rad(G)-1$.
Therefore, $a \in A$ and $b \in B$.
\end{proof}

\begin{corollary}
%Let $G$ be a distance-hereditary graph.
If $diam(G) \geq 2rad(G) - 1$ then $C(G)$ is a cograph.
\end{corollary}

It remains now to investigate the case when $diam(G)=2rad(G)-2$. 

\begin{lemma}\label{lem:hellyLikeCenters}
Let $diam(G) = 2rad(G) - 2$,
%Let $G$ be a distance-hereditary graph with $diam(G) = 2rad(G) - 2$
and let $M \subseteq C(G)$.
If all $u,v \in M$ satisfy $d(u,v) = 2$, then there is a vertex $c \in C(G)$ that is universal to $M$. %\todo{used only in following lemma - should we combine?}
\end{lemma}
\begin{proof}
Consider a disk of radius 1 centered at each $s \in M$ and a disk of radius $rad(G) - 1$ centered at each $v \in V \setminus M$.
Any two vertices $u,v \in V \setminus M$ satisfy $d(u,v) \leq diam(G) = 2rad(G) - 2 = r(u) + r(v)$.
Since $M \subseteq C(G)$, any $s \in M$ and $v \in V$ satisfy $d(s,v) \leq rad(G) = r(s) + r(v)$.
By assumption, any two $s,t \in M$ satisfy $d(s,t) = 2 = r(s) + r(t)$.
%Hence, the disks pairwise intersect.
By Proposition~\ref{prop:dominatingClique}, there is a single vertex or a pair of adjacent vertices $r$-dominating $G$.
In the former case, we are done.
Thus, consider the case when there is an $r$-dominating edge $ab \in E$.
We have $a,b \in C(G)$ since all vertices~$v \in V$ see some  end-vertex of edge $ab$ within $rad(G) - 1$.
We claim that at least one end-vertex of edge $ab$ is universal to~$M$.
By contradiction assume there exist vertices $u,v \in M$ which are adjacent to opposite ends of edge $ab$.
Without loss of generality, let $u \in N(a) \setminus N(b)$ and $v \in N(b) \setminus N(a)$.
Since $d(u,v) = 2$, we get in $G$ either an induced $C_5$, or an induced house, or an induced gem. A contradiction obtained proves the lemma.
%then there must exist a chord $ub \in E$ or $va \in E$, either of which contradict the choice of~$u$ and~$v$.
\end{proof}

%%%%%%%%%%%%%%%%%%%%%%%%%%%%%%%%%%%%%%%%%%%%%%%%%%%%%%%%%%%%%%%
% if d=2r-2 then C(G) = A u S u B u C,
%                C(G) is either cograph or d(C(G)) = 3, 
%                D(S-,1) includes C(G)
%%%%%%%%%%%%%%%%%%%%%%%%%%%%%%%%%%%%%%%%%%%%%%%%%%%%%%%%%%%%%%%
\begin{lemma}\label{lem:centerWhen2r-2}
Let $x,y$ be a diametral pair of $G$, $diam(G) = 2rad(G) - 2$,
%Let $G$ be a distance-hereditary graph with $diam(G) = 2rad(G) - 2$ and a diametral pair $x,y$,
and let
$A \coloneqq S_{rad(G)-2}(x,y)$, $S \coloneqq S_{rad(G)-1}(x,y)$, and $B \coloneqq S_{rad(G)-2}(y,x)$.
Then $A \cup B \cup (S \cap C(G)) \subseteq C(G)$ and there is a vertex $c \in S \cap C(G)$.
Moreover, $C(G) \subseteq D(S \cap C(G), 1)$.
\end{lemma}
\begin{proof}
Consider any $s \in A$. By Lemma~\ref{lem:duality1}, $e(s) \leq \max\{\max\{d(s,x), d(s,y)\}, \min\{d(s,x), d(s,y)\} + 2\} = d(s,y) = rad(G)$.
Hence, $s \in C(G)$ and, by symmetry, $A \cup B \subseteq C(G)$.

By contradiction assume there is no central vertex in $S$.
Let $w$ be a vertex from $S$ minimizing $|F(w)|$, and let $v \in F(w)$.
Since $w \notin C(G)$, $d(w,v) \geq rad(G) + 1$.
Denote by $s_1 \in S_1(w,x)$ and $s_2 \in S_1(w,y)$ two adjacent to $w$ vertices on a shortest path from $x$ to $y$. 
As previously established, both $s_1$ and $s_2$ are central since they belong to $A$ and $B$, respectively.
Thus, $d(s_1,v) \leq rad(G)$ and $d(s_2,v) \leq rad(G)$.
Therefore, $d(s_1,v) = d(s_2,v) = rad(G)$ and $d(w,v) = rad(G) + 1$.
Since $s_1,s_2 \in N^{rad(G)}(v)$ and are connected via~$w$ in the graph $< V \setminus N^{rad(G)-1}(v) >$, by Proposition~\ref{prop:dhgCharacterization}\ref{byDownNeighbors}, there is a vertex $t \in N^{rad(G)-1}(v)$ adjacent to~$s_1$ and~$s_2$.
As $t \in S_{rad(G)-1}(y,x)$ and $v \in F(w) \setminus F(t)$,
by minimality of $|F(w)|$, there is a vertex $u \in F(t) \setminus F(w)$.
By our assumption, $u\notin C(G)$, i.e.,  $d(u,t) \geq rad(G) + 1$.
Consider the \FPC{} on vertices $t,w,v,u$. We have 
 $d(v,u) + d(w,t) = d(v,u) + 2$,
 $d(v,t) + d(u,w) \leq rad(G) - 1 + rad(G)$, and
 $d(v,w) + d(u,t) \ge rad(G) + 1 + rad(G) + 1 = 2rad(G) + 2$.
Since the latter two sums differ by more than 2, necessarily, $d(v,u) + d(w,t) = d(v,w) + d(u,t)$.
Hence, $d(v,u) = d(v,w) + d(u,t) - d(w,t) \ge 2rad(G)$, a contradiction with $diam(G) = 2rad(G) - 2$.
Thus, there is a vertex $c \in S \cap C(G)$.

Next, we establish an intermediate claim that $C(G) \subseteq D(M,1)$, where $M \coloneqq A \cup B \cup (S \cap C(G))$.
By contradiction suppose there is a vertex $w \in C(G)$ with $w \notin D(M,1)$.
Consider arbitrary vertices $a \in A$ and $b \in B$.
Thus, $d(a,b)=2$, and $a,b \in M$ and, by the choice of~$w$, necessarily $d(w,a) \geq 2$ and $d(w,b) \geq 2$.
If $d(w,a)=d(w,b)=2$ then, by Lemma~\ref{lem:hellyLikeCenters} applied to the set $\{w,a,b\}$, there is a central vertex~$u$ adjacent to $w,a,b$.
In this case $u \in S \cap C(G)$ and therefore $u \in M$, contradicting with $w \notin D(M,1)$.
Assume now, without loss of generality, that $d(w,a) \geq 3$. 
%We obtain a general contradiction.
Consider the \FPC{} on vertices $w,x,y,a$.
We have $d(x,y) + d(w,a) \geq 2rad(G) + 1$ is the largest sum
since $d(x,w) + d(y,a) \leq 2rad(G)$
and $d(x,a) + d(w,y) \leq 2rad(G) - 2$.
Since the smaller two sums must be equal and differ from the larger one by at most two, inequality $d(x,y) + d(w,a) \geq d(x,a) + d(y,w) + 3$ gives a contradiction which establishes the claim that $C(G) \subseteq D(M,1)$.

Finally, we establish that $C(G) \subseteq D(S \cap C(G), 1)$.
By contradiction assume there is a central vertex $w \in C(G) \setminus S$ which is not adjacent to any vertex of $S \cap C(G)$.
By the previous claim, $w$ is adjacent to some vertex from $A$ or $B$.
Without loss of generality, let $wa \in E$ for some vertex $a \in A$.
Since $d(a,y) = rad(G)$, necessarily, $d(w,y) \geq rad(G) - 1$.
If $d(w,y) = rad(G) - 1$ then $w \in S$, a contradiction.
So, $d(w,y) = rad(G)$ must hold.
Now, vertices~$w$ and~$a$ are connected in $< V \setminus N^{rad(G)-1}(y) >$.
By Proposition~\ref{prop:dhgCharacterization}~\ref{byDownNeighbors}, $N(w) \cap N^{rad(G)-1}(y) = N(a) \cap N^{rad(G)-1}(y)$.
By Proposition~\ref{prop:slicesAreJoined}, also $S \cap C(G) \subseteq N(a) \cap N^{rad(G)-1}(y)$.
Thus, $w$ is universal to $S \cap C(G)$, a contradiction.
\end{proof}

We are ready to prove the main result of this section. 
%%%%%%%%%%%%%%%%%%%%%%%%%%%%%%%%%%%%%%%%%%%%%%%%%%%%%%
% centers characterization
%%%%%%%%%%%%%%%%%%%%%%%%%%%%%%%%%%%%%%%%%%%%%%%%%%%%%%
\begin{theorem}
Let $H$ be a subgraph of a distance-hereditary graph $G$ induced by $C(G)$.
Either \\
(i) $H$ is a cograph, or\\
(ii) $H$ is a connected distance-hereditary graph with $diam(H)=3$ and $C(H)$ is a connected cograph with $rad(C(H))=2$.\\
Furthermore, any such graph $H$ is the center of some distance-hereditary graph.
%Furthermore, any cograph and any connected distance-hereditary graph with $diam(H)=3$ and center $C(H)$ as a connected cograph with $rad(C(H))=2$ is the center of some distance-hereditary graph.
\end{theorem}
\begin{proof}
If $diam(G) \geq 2rad(G) - 1$ then, by Lemma~\ref{lem:centerWhen2r} and Lemma~\ref{lem:centerWhen2r-1}, $H$ is a cograph.
Assume now that $diam(G) = 2rad(G) - 2$ and $diam(H) = 3$ (if $diam(H) \leq 2$ then, by definition, $H$ is a cograph). 
Then, $H$ is a connected distance-hereditary graph~\cite{YEH2003297}, and so $2rad(H) - 2 \leq diam(H) \leq 2rad(H)$. %, where $rad(H)$ is an integer.
On one hand, $rad(H) \geq \lceil (diam(H)/2) \rceil = 2$.
On the other hand, $rad(H) \leq \lfloor (diam(H) + 2)/2 \rfloor = 2$.
Hence, $rad(H) = 2$.

So, $H$ is a connected distance-hereditary graph with $diam(H)=3$ and $rad(H)=2$. Consider the center $C(H)$ of $H$. 
First we show that $C(H)$ is connected. 
Let $r(u)=1$ for each vertex~$u \in H$.
Then, any pair $u,v \in H$ satisfies $d_H(u,v) \leq diam(H) = r(u) + r(v) + 1$.
By Proposition~\ref{prop:dominatingClique}, there is a clique~$K$ in $H$ dominating $H$.
Since each vertex of~$K$ is at most distance $2=rad(H)$ from every vertex  of~$H$, $K \subseteq C(H)$ holds. 
Moreover, %any vertex of $H$ is adjacent to some vertex of $K$. So,
every two vertices of $C(H)$ are connected through vertices of $K\subseteq C(H)$, implying that $C(H)$ is connected in $H$.
In distance-hereditary graphs every connected subgraph is isometric. Hence, $C(H)$ is an isometric subgraph of $H$. As $rad(H)=2$, every two vertices of $C(H)$ are at distance at most 2 from each other, implying  $diam(C(H))=2$. Thus, $C(H)$ is a connected cograph with $rad(C(H))\leq 2$. 
%------------------------------
%We first claim that $C(H)$ does not induce a star
%if and only if $C(H)$ is a connected cograph of radius 2.
%If $C(H)$ is a connected cograph of radius 2, then $C(H)$ does %not induce a star since the radius of a star is 1.
%If $C(H)$ does not induce a star, then $rad(C(H)) \geq 2$.
%However $rad(C(H)) \leq rad(H) = 2$.
%Hence $C(H)$ is self-centered, wherein all vertices of $C(H)$ %are at most distant 2 and $C(H)$ has radius 2.
%Since $C(H)$ is also distance-hereditary, there are no induced %paths of length 3 or more, and therefore $C(H)$ is a cograph.
%------------------ 
%It remains only to show that $C(H)$ is connected.
%Let $r(u)=1$ for each vertex~$u \in H$.
%Then any $u,v \in H$ satisfies $d(u,v) \leq diam(H) = r(u) + %r(v) + 1$.
%By Proposition~\ref{prop:dominatingClique} there is an %r-dominating clique~$K \subseteq H$.
%Since each vertex of~$K$ is at most distance 2 from all %vertices of~$H$, including $C(H)$, then $K \subseteq C(H)$.
%Moreover, any two vertices of $C(H)$ are adjacent to vertices %of clique $K$,
%and therefore are connected by an induced (shortest) path.
%Thus, if $C(H)$ does not induce a star, then $C(H)$ is a %connected cograph of radius 2.
%-----------------------------

We will show next that $rad(C(H))= 2$, i.e.,  
for any $c \in C(H)$ there is a vertex $z \in C(H)$ such that $cz \notin E$. 
Consider a vertex~$t \in F(c)$ furthest from $c$ in $G$. We have $d_G(c,t)=rad(G)$. 
Let $z \in H$ be a closest vertex to~$t$ which is central in $G$.
Since $diam(G)=2rad(G)-2$, by Lemma~\ref{lem:eccFunctionEqn}, we have $d_G(t,z) = d_G(t,C(G)) = e_G(t) - rad(G) \leq diam(G) - rad(G) = rad(G) - 2$.
Moreover, vertices~$z$ and~$c$ are not adjacent since $d_G(c,t)=rad(G)$ and $d_G(t,z) \leq rad(G) - 2$.
But, since $c \in C(H)$, $d_G(c,z)= d_H(c,z)\leq 2$.
Therefore, $d_H(c,z) = 2$ and $d_G(t,z) = rad(G) - 2$.
We next establish that $z$ belongs to $C(H)$.
By contradiction, assume that there is a vertex~$u \in H$ such that $d_H(z,u) > rad(H) = 2$.
Then, $d_H(z,u) = diam(H) = 3$ and, by the choice of~$c$ ($c\in C(H)$), necessarily $d_H(c,u) \leq 2$.
% We had used down-neighbors with 3 cases based on where u belongs....
%Since $d(c,t)=rad(G)$ then vertex~$u$ belongs to $N^{rad(G)}(t)$, $N^{rad(G)-1}(t)$, or $N^{rad(G)-2}(t)$.
% Or, we could use 4PC...
Consider the \FPC{} on vertices $c,u,z,t$.
We have that $d(c,t) + d(u,z) = rad(G) + 3$ is the largest sum
since $d(c,z) + d(u,t) \leq rad(G) + 2$
and $d(c,u) + d(z,t) \leq rad(G)$.
However, $d(c,t) + d(u,z) \geq d(c,u) + d(z,t) + 3$, giving a contradiction since the smaller two sums must be equal and differ from the larger one by at most two. 
Hence, $z$ belongs to $C(H)$ showing that every $c \in C(H)$ has a non-adjacent vertex $z \in C(H)$. 
%Therefore $C(H)$ does not induce a star.

Finally, we show that any such graph $H$ is the center of some distance-hereditary graph $G$.  In what follows, we refer to Figure ~\ref{fig:centerCharacterization} for an illustration.
If $H$ is a cograph, then one can construct a graph $G$ by simply adding to $H$ four new vertices $x,x^*,y,y^*$. Vertices ~$x$ and~$y$ are universal to~$H$, and vertices $x^*,y^*$ are pendant to~$x$ and~$y$, respectively. Now graph $H$ is the center of $G$ as any vertex~$u$ of the cograph~$H$ is at most distance 2 to any vertex of~$G$, whereas $d_G(x,y^*)=3$ and $d_G(y,x^*)=3$.
Suppose now that $H$ is a connected distance-hereditary graph with $diam(H)=3$ and $C(H)$ is a connected cograph with $rad(C(H))=2$. 
One can construct a graph $G$ by adding to $H$ (with $C(H) = \{c_1,c_2,...,c_\ell\}$) $\ell$ new vertices $x_1,x_2,...,x_\ell$ such that each $x_i$ is pendant to $c_i \in C(H)$.
Each $c_i \in C(H)$ has $d_G(c_i,u) \leq 2$ for all~$u \in H$.
Since $rad(C(H))=2$, each $c_i$ has a non-adjacent vertex $c_k \in C(H)$, and therefore $d_G(c_i,x_k) = 3$.
Any vertex~$u \in H \setminus C(H)$ has a vertex~$v \in H \setminus C(H)$, for which $d_G(u,v) = 3$. Furthermore, any such $u$ satisfies $d_G(u,x_i) = d_G(u,c_i) + 1 \leq 3$ for each $x_i$.
Since $rad(C(H))=2$, for any pendant~$x_i$, vertex~$c_i$ has a non-adjacent vertex~$c_k \in C(H)$ and therefore $d_G(x_i,x_k) = 4$.
Hence,  $H$ is the center of $G$.
\end{proof}

\begin{figure}[H]
  \begin{center}
    \includegraphics[scale=0.7]{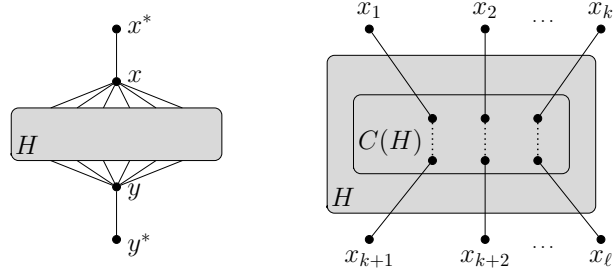}
	\caption{Any cograph $H$ (left), and any connected distance-hereditary graph $H$ with diameter 3 where $C(H)$ is a connected cograph with radius 2 (right), is the center of some distance hereditary graph.}
    \label{fig:centerCharacterization}
  \end{center}
\end{figure}

%%%%%%%%%%%%%%%%%%%%%%%%%%%%%%%%%%%%%%%%%%%%%%%%%%%%%%
%~~~~~%~~~~~~~%~~~~~%~~~~~~~%~~~~~%~~~~~~~%~~~~~%~~~~%
%~~~~~%~~~~~~~%~~~~~%~~~~~~~%~~~~~%~~~~~~~%~~~~~%~~~~%
%%%%%%%%%%%%%%%%%%%%%%%%%%%%%%%%%%%%%%%%%%%%%%%%%%%%%%

\section{Computing all eccentricities}\label{sec:allEccentricities}
{Recall that Theorem~\ref{thm:allBounds} yields a linear time 2-approximation of all eccentricities. We also used ideas from Section~\ref{sec:centers} to find in linear time a set $S \subseteq C(G)$ sufficient for an additive 1-approximation of all eccentricities. However, it proved to be challenging to get exact computation of all eccentricities using that approach. Instead, we found it more fruitful to use a characteristic pruning sequence of~$G$.}
We provide in this section a simple linear time algorithm to compute all eccentricities of a distance-hereditary graph $G$ using a weight function and a special pruning sequence produced by processing layers of a breadth-first search tree of $G$ (see~\cite{DAMIAND200199}).

For a given graph $G=(V(G),E(G))$ and an $n$-tuple $(p(v_1), p(v_2), ..., p(v_n))$ of non-negative vertex weights% for each $v_i \in V(G)$
, we define the $p$-weighted eccentricity of each vertex $v \in V(G)$ as 
$e_{G,p}(v) = \max_{u \in V(G)}\{ d_G(v,u) + p(u) \}$.
We refer to the set of furthest vertices from a vertex~$v$ under weight function $p$ in $G$
as $F_{G,p}(v) = \{ u \in V(G): e_{G,p}(v) = d_G(v,u) + p(u) \}$. Clearly, when $p(v)=0$ for all $v\in V(G)$, we have  $e_{G,p}(v) =e_G(v)$ and $F_{G,p}(v)=F_G(v)$ agreeing with earlier definitions.  

\begin{lemma}\label{lem:pruningTwins}
Let $x,y \in V$ be twins with $p(x) \leq p(y)$.
Set $G' = G - \{x\}$ and $p'(v) = p(v)$ for all $v \in V(G')$.
Then, $e_{G,p}(v) = e_{G',p'}(v)$ for all $v \in V(G') \setminus \{y\}$,
 $e_{G,p}(y) = \max\{p(x) + d_G(x,y), e_{G',p'}(y)\}$,
 and $e_{G,p}(x) = \max\{p(y) + d_G(x,y), e_{G,p}(y)\}$.
 %%%%%
Moreover, if $F_{G,p}(y) \setminus \{x\} \neq \emptyset$, then $e_{G,p}(y) = e_{G',p'}(y)$.

%%%%%
%Moreover, if $p(x) < rad(G) - d_G(x,y)$, then $e_{G',p'}(y) = e_{G,p}(y)$
%\setlist{nolistsep}
%\begin{enumerate}[noitemsep, label=(\roman*)]
  %\item $e_{G,p}(v) = e_{G',p'}(v)$ for all $v \in V' - \{y\}$,
  %\item $e_{G,p}(y) = \max\{p(x) + d_G(x,y), e_{G',p'}(y)\}$, and
  %\item $e_{G,p}(x) = \max\{p(y) + d_G(x,y), e_{G,p}(y)\}$.
%\end{enumerate}
\end{lemma}
\begin{proof}
Let $v \in V(G') \setminus \{y\}$.
As $x$ and $y$ are twins, $d_G(v,y) = d_G(v,x)$.
Since $G'$ is isometric in $G$, $d_{G}(v,u) + p(u) = d_{G'}(v,u) + p'(u)$ for any $u \in V(G')$.
Then, by definition of eccentricity, $e_{G',p'}(v) = \max_{u \in V(G')}\{ d_{G'}(v,u) + p'(u)\} = \max_{u \in V(G)\setminus \{x\}}\{ d_{G}(v,u) + p(u)  \}$.
As $p(x) \leq p(y)$, $d_G(v,y) + p(y) \geq d_G(v,x) + p(x)$.
Thus, $e_{G',p'}(v) = \max\{ d_G(v,x) + p(x), \max_{u \in V(G) \setminus  \{x\}}\{ d_{G}(v,u) + p(u) \} \} 
                    = \max_{u \in V(G)}\{ d_{G}(v,u) + p(u) \} 
                    = e_{G,p}(v) $.
Similarly, %by definition, % and since $G'$ is isometric in $G$,
$e_{G,p}(y) %= \max_{u \in V}\{ d_G(y,u) + p(u) \} 
            = \max\{ d_G(x,y) + p(x), \max_{u \in V(G')}\{ d_G(y,u) + p(u) \}  \}
            = \max\{ d_G(x,y) + p(x), e_{G',p'}(y) \}$.
%%%%%
Moreover, if $F_{G,p}(y) \setminus \{x\}\neq\emptyset$, 
then $e_{G,p}(y) = e_{G',p'}(y)$ as realized by the weighted distance in $G'$ from $y$ to any $v \in F_{G,p}(y) \setminus \{x\}$.
%
%%%%%

As $d(x,y) \geq 1$ and $p(x) \le p(y)$, $d_G(x,y) + p(y) \geq \max\{d_G(x,y) + p(x), p(y)\}$.
Again, by definition of eccentricity,
\vspace*{-0.2cm}
\begin{align*}
e_{G,p}(x) &= \textstyle{\max\{ p(x), d_G(x,y) + p(y), \max_{u \in V(G) \setminus  \{x,y\}} \{ d_G(x,u) + p(u)\}\}} \\
           % &= \max\{ d_G(x,y) + p(y), \max_{u \in V - \{x,y\}} \{ d_G(y,u) + p(u)\}\} \\
            &= \textstyle{\max\{ d_G(x,y) + p(y), \max_{u \in V(G) \setminus  \{x,y\}} \{ d_G(y,u) + p(u)\}, d_G(x,y) + p(x), p(y) \}} \\
            %&= \max\{ d_G(x,y) + p(y), \max_{u \in V} \{ d_G(y,u) + p(u)\} \}\\
            &= \max\{ d_G(x,y) + p(y), e_{G,p}(y) \}. \qedhere
\end{align*}
%Consider now $x \in \layer_i$ and any $v \in \layer_{i+1} \cap V'$ adjacent to $y$.
%By contradiction, assume that $F_{G',p'}(y) = \{v\}$.
%Hence, $e_{G',p'}(y) = d_{G'}(y,u) + p'(u) = 1 + p'(u)$.
\end{proof}

\begin{lemma}\label{lem:pruningPendants}
Let $x$ be a vertex pendant to $y$.
Set $G' = G - \{x\}$, 
$p'(y) = \max\{1 + p(x), p(y)\}$, and $p'(v) = p(v)$ for all $v \in V(G') \setminus \{y\}$.
Then,
$e_{G,p}(v) = e_{G',p'}(v)$ for all $v \in V(G')$.
If $F_{G,p}(y) \setminus \{x\} \neq \emptyset$, then $e_{G,p}(x) = \max\{p(x), e_{G',p'}(y) + 1\}$.
%If $p(x) < rad(G) - 1$, then $e_{G,p}(x) = \max\{p(x), e_{G',p'}(y) + 1\}$.
\end{lemma}
\begin{proof}
Let $v \in V(G')$.
As $x$ is pendant to $y$,  $d_G(v,x) = d_G(v,y) + 1$.
Since $G'$ is isometric in $G$, $d_{G'}(v,y) + p'(y) = \max\{d_G(v,y) + 1 + p(x), d_G(v,y) + p(y)\} =  \max\{d_G(v,x) + p(x), d_G(v,y) + p(y)\}$.
Then, by definition of eccentricity,
\vspace*{-.25cm}
\begin{align*}
e_{G',p'}(v) &= \textstyle{\max\{ d_{G'}(v,y) + p'(y),  \max_{u \in V(G') \setminus  \{y\}}\{ d_{G'}(v,u) + p'(u) \}  \} } \\
              &= \textstyle{\max\{ d_G(v,x) + p(x), d_G(v,y) + p(y),  \max_{u \in V(G') \setminus  \{y\}}\{ d_{G}(v,u) + p(u) \}  \} } \\
              &= \textstyle{\max_{u \in V(G)}\{d_G(v,u) + p(u) \} } 
              = e_{G,p}(v).
\end{align*}

Assume now that $F_{G,p}(y) \setminus \{x\} \neq \emptyset$.
Hence, $e_{G',p'}(y) = e_{G,p}(y) = \max_{u \in V(G) \setminus  \{x\}}\{d_G(y,u) + p(u)\}$.
Thus, by definition of eccentricity, 
\vspace*{-.25cm}
\begin{align*}
e_{G,p}(x) &= \textstyle{ \max\{ p(x), p(y) + 1, \max_{u \in V(G) \setminus  \{x,y\}}\{ d_G(x,u) + p(u) \}  \}  } \\
           &= \textstyle{ \max\{ p(x), p(y) + 1, \max_{u \in V(G) \setminus  \{x,y\}}\{ d_G(y,u) + p(u) \} + 1  \}  } \\
           &= \textstyle{ \max\{ p(x), \max_{u \in V(G) \setminus  \{x\}}\{ d_G(y,u) + p(u) \} + 1  \}  }
           = \textstyle{ \max\{ p(x), e_{G',p'}(y) + 1 \}.  } \qedhere
\end{align*}
\end{proof}

We use the pruning sequence (the vertex elimination ordering) $\sigma = (v_1,...,v_n)$ that can be constructed in linear time via $\rho$ iterations of a systematic removal of pendants/twins from each layer $\layer_{\rho},...,\layer_1$ of a breadth-first search tree rooted at $v_n$ (see~\cite{DAMIAND200199}).
Iteration $k$, where $k=\rho,...,1$, consists of four consecutive steps:
\setlist{nolistsep}
\begin{itemize}[noitemsep]
\item[(a)] remove any $x \in \layer_k$ twin to some $y \in \layer_k$ of the same connected component in $\layer_k$ (i.e., $x$ and $y$ belong to the same connected component of the subgraph of $G$ induced by vertices of $\layer_k$),
\item[(b)] remove any $x \in \layer_k$ pendant to some $y \in \layer_{k-1}$,
\item[(c)] remove any $x \in \layer_{k-1}$ twin to some $y \in \layer_{k-1}$ belonging to the same neighborhood $N(z) \cap \layer_{k-1}$ of some $z \in \layer_k$, and 
\item[(d)] remove any $x \in \layer_k$ pendant to some $y \in \layer_{k-1}$.
\end{itemize}
Note that we move to the next step only when no vertex remains satisfying the condition of the previous step. 
At the end of iteration $k$, all vertices of $\layer_k$ have been removed (see~\cite{DAMIAND200199}).
By this ordering, any $u \in \layer_k$ satisfies that
if $u$ is a pendant to $v$, then $v \in \layer_{k-1}$, and
if $u$ is a twin to $v$, then $v \in \layer_k$.
Let $G_i$ denote the graph induced by $\{v_i,...,v_n\}$ for each $i=1,...,n$.

\begin{theorem}
There is a linear time algorithm to compute all eccentricities in a distance-hereditary graph.
\end{theorem}
\begin{proof}
%%%% Elimination ordering description and properties
%There is a vertex elimination ordering $\sigma = (v_1,...,v_n)$ (see~\cite{DAMIAND200199}) that can be constructed in linear time via $\rho$
%iterations of a systematic removal of pendants/twins from each layer $\layer_{\rho},...,\layer_0$ of a breadth-first search tree rooted at $v_n$.
%Iteration $k$ consists of four steps:
%(a) remove any $x \in \layer_k$ twin to some $y \in \layer_k$ of the same connected component in $\layer_k$,
%(b) remove any $x \in \layer_k$ pendant to some $y \in \layer_{k-1}$,
%(c) remove any $x \in \layer_{k-1}$ twin to some $y \in \layer_{k-1}$ belonging to the same neighborhood $N(z) \cap \layer_{k-1}$ of some $z \in \layer_k$, and 
%(d) remove any $x \in \layer_k$ pendant to some $y \in \layer_{k-1}$.
%At the end of iteration $k$, all vertices of $\layer_k$ have been removed.
%By this ordering, any $u \in \layer_k$ satisfies that
%if $u$ is a pendant to $v$, then $v \in \layer_{k-1}$, and
%if $u$ is a twin to $v$, then $v \in \layer_k$.
%Let $G_i \coloneqq G(\{v_i,...v_n\})$ for each $i=1,...,n$.

%%% Algorithm idea
Let $\sigma = (v_1,...,v_n)$ be the pruning sequence constructed as described above by
each iteration $k=rad(G),...,1$ of removing vertices from layer $\layer_k$ of a BFS tree rooted at a central vertex $v_n$. %~\cite{DAMIAND200199}.
%For each $v_i$, initialize $p(v_i) = 0$.
%We will complete each iteration $k=rad(G),...,2$ of the pruning sequence. %, then partially complete iteration $i=2$ executing steps $a,b,c$
Denote by $v_y \in \layer_2$ the first pendant vertex of $\sigma$ encountered in step (d)  of iteration 2 (or in step (b) if $\layer_2$ becomes empty after steps (a) and (b)). Denote by $v_z$ the first vertex of $\sigma$ encountered in iteration 1.
Thus, the graph $G_z$ consists of $v_n$ and some  twins/pendants in  $\layer_1$ adjacent to~$v_n$.
The algorithm is summarized as follows. %by the following steps.
%Initially, each vertex $v_i$, $i < z$, is processed along a forward direction of $\sigma$. 
We process vertices $v_i$, $i < z$, from $v_1$ to $v_z$ (from left to right along $\sigma$).  
We denote by $p_i$ the weight function of each vertex immediately before vertex~$v_i$ is processed. For each vertex $v_j \in \sigma$, set $p_1(v_j) = 0$.
As each $v_i$ is processed, $p_{i+1}$ is invariant (that is, $p_{i+1}(v_j) = p_{i}(v_j)$ for every $v_j$) with the exception of one case: if $v_i$ is a pendant to $v_j$ in $G_i$, then let $p_{i+1}(v_j) = \max\{p_{i}(v_i) + 1, p_{i}(v_j)\}$.
We can assume that if $v_i$ is a twin to vertex $v_j$ in $G_i$, then $p_i(v_i) \le p_i(v_j)$, since otherwise, their positions as twins can be swapped in $\sigma$.
Observe that
the weight function of a vertex $v_j \in \layer_k$ can only increase when a vertex $v_i \in \layer_{k+1}$ is pendant to $v_j$, where $i < j$.
Hence, if any vertex $v_i$ belongs to layer $\layer_k$, then for any integer $\ell$, $p_\ell(v_i) \leq rad(G) - k$.
Additionally, every vertex $v_i$ and integers $\ell < \kappa$ satisfy $p_\ell(v_i) \le p_\kappa(v_i)$.
After all vertices $v_i$, $i < z$, are processed (along $\sigma$ from left to right), we next compute all $p_z$-weighted eccentricities in $G_z$.
Then, we compute all $p_y$-weighted eccentricities in $G_y$.
Finally, we process each vertex $v_i$, $i < y$, along the reverse direction of $\sigma$. The $p_{i+1}$-weighted eccentricities in $G_{i+1}$ are used to obtain the $p_i$-weighted eccentricities in $G_i$.

%

%The remaining graph $G_z$ consists of $v_n$ and twins/pendants of $\layer_1$ adjacent to~$v_n$.
\emph{Backward phase 1: Compute all $p_z$-weighted eccentricities in $G_z$.}
Denote by $V^*$ and $N_i^*$ the vertex lists ordered by decreasing weight $p_z$ from the respective vertex sets $V(G_z)$ and $N(v_i)\cap V(G_z)$ for each $v_i \in V(G_z)$.
The lists are ordered in total linear time with a bucket sort.
The first vertex $w$ of $N_{i}^*$ has maximum $p_z(w)$ among neighbors of $v_i$ in $G_z$, and 
the first vertex $u$ of $V^*\setminus N_i^*$ has maximum $p_z(u)$ among non-neighbors of $v_i$.
By definition, the weighted eccentricity of each $v_i$ is $e_{G_z,p_z}(v_i)% = \max\{p_z(v_i),  1 + \max_{u \in N(v_i)}p_z(u),  2 + \max_{u \in V(G_z) \setminus N[v_i]}p_z(u) \}
= \max\{p_z(v_i), 1 + p_z(w), 2 + p_z(u)\}$
if there exists a non-neighbor $u \in V^* \setminus N_i^*$, and
$e_{G_z,p_z}(v_i) = \max\{p_z(v_i), 1 + p_z(w)\}$ otherwise.

\emph{Backward phase 2: Compute all $p_y$-weighted eccentricities in $G_y$.}
By choice of $y$, each vertex $v_i$ for $y \le i < z$ satisfies $v_i \in \layer_2$ and $v_i$ is pendant to a vertex $v_j \in \layer_1$ in $G_y$.
%Any vertex $v_\ell \in \layer_1$ has a larger value weight in $p_z$ than $p_y$ only if there is a vertex $v_i$ pendant to $v_\ell$.
Hence, for $M = \{v_z,...,v_n\}$ and $S = \{ v_y,...,v_{z-1} \}$, $\max_{m \in M}p_z(m) = \max\{ \max_{m \in M}p_y(m),  \max_{u \in S}p_y(u) + 1  \}$ holds.
Therefore, we may again use ideas from the previous phase.
For each vertex $v_i$, $y \le i < z$, pendant to $v_j$, we have a vertex $w \in N(v_j)$ in $G_z$ with maximum $p_z(w)$ among neighbors of $v_j$ in $G_z$ and a vertex $u \notin N(v_j)$ in $G_z$ with maximum $p_z(u)$ among non-neighbors of $v_j$ in $G_z$.
By definition, the weighted eccentricity of each $v_i$ is
$e_{G_y,p_y}(v_i) = \max\{
p_i(v_i),
p_i(v_j) + 1,
p_z(w) + 2,
p_z(u) + 3
\}$ if there exists in $G_z$ a non-neighbor $u \notin N(v_j)$, and
$e_{G_y,p_y}(v_i) = \max\{
p_i(v_i),
p_z(v_j) + 1,
p_z(w) + 2
\}$ otherwise.
Additionally, by Lemma~\ref{lem:pruningPendants}, each vertex $v_\ell$ for $\ell \ge z$ has equal $p_z$ and $p_y$ weighted eccentricities, that is, $e_{G_y,p_y}(v_\ell) = e_{G_z,p_z}(v_\ell)$.

\emph{Backward phase 3: Compute all $p_i$-weighted eccentricities in $G_i$ for $i < y$ along a reverse direction of $\sigma$.}
If $v_i$ is a twin to $v_j$ in $G_i$, by Lemma~\ref{lem:pruningTwins},
$e_{G_i,p_i}(v_j) = \max\{p_i(v_i) + d_G(v_i,v_j), e_{G_{i+1},p_{i+1}}(v_j)\}$,
$e_{G_i,p_i}(v_i) = \max\{p_i(v_j) + d_G(v_i,v_j), e_{G_i,p_i}(v_j)\}$,
and for all $u \in V(G_i) \setminus \{v_i, v_j\}$
$e_{G_i,p_i}(u) = e_{G_{i+1},p_{i+1}}(u)$.
We consider now the case that $v_i$ is a pendant to $v_j$.
By Lemma~\ref{lem:pruningPendants}, each $u \in V(G_i) \setminus \{v_i\}$ satisfy
$e_{G_i,p_i}(u) = e_{G_{i+1},p_{i+1}}(u)$.
We claim that $e_{G_i,p_i}(v_i) = \max\{p_i(v_i), e_{G_{i+1},p_{i+1}}(v_j) + 1\}$.
It remains only to show that any pendant $v_i$, where $i < y$, satisfies $F_{G_i,p_i}(v_j) \setminus \{v_i\} \neq \emptyset$; applying Lemma~\ref{lem:pruningPendants} then proves the claim.

By contradiction, let $v_i \in \layer_\gamma$ pendant to $v_j \in \layer_{\gamma-1}$ be the earliest vertex in $\sigma$ with $i < y$ and $F_{G_i,p_i}(v_j) = \{v_i\}$.
Hence, $e_{G_i,p_i}(v_j) = d_{G_i}(v_j,v_i) + p_i(v_i) > \max_{i+1 \le t \le n}\{d_{G_i}(v_j,v_t) + p_i(v_t)\}$.
As $i < y$ and $v_y \in \layer_2$, then $\gamma \ge 2$.
Thus, $p_i(v_i) \le rad(G) - 2$ and $e_{G_i,p_i}(v_j) = d_G(v_j,v_i) + p_i(v_i) \le rad(G) - 1$.
If $e_{G_i,p_i}(v_j) = e_{G}(v_j)$, we obtain a contradiction with $e_{G}(v_j) \le rad(G) - 1$.
Therefore, $e_{G_i,p_i}(v_j) < e_{G}(v_j)$.
Let $v_\ell \in \sigma$ be the earliest vertex such that $e_{G_\ell,p_\ell}(v_j) > e_{G_{\ell+1},p_{\ell+1}}(v_j)$, where $\ell < i < j$.
By Lemma~\ref{lem:pruningTwins} and Lemma~\ref{lem:pruningPendants},
$v_\ell$ is a twin to $v_j$ in $G_\ell$ such that $F_{G_\ell,p_\ell}(v_j) = \{v_\ell\}$ and $p_\ell(v_\ell) \le p_\ell(v_j)$.
Then, $rad(G) \le e_G(v_j) = e_{G_\ell,p_\ell}(v_j) = d_G(v_\ell,v_j) + p_\ell(v_{\ell})$.
Hence, $p_i(v_j) \ge p_\ell(v_j) \ge p_\ell(v_{\ell})\ge rad(G) - d_G(v_\ell,v_j)$.
If $p_i(v_j)=rad(G) - 1$ or if $v_\ell$ is a true twin to $v_j$, then $p_i(v_j) \geq rad(G) - 1$, and a contradiction arises with
$p_i(v_j) < d_G(v_j,v_i) + p_i(v_i) \le rad(G) - 1$ (recall that $F_{G_i,p_i}(v_j) = \{v_i\}$).
Necessarily, $v_\ell$ is a false twin to $v_j$ and $p_i(v_j) = rad(G) - 2$.
As $rad(G) - 2 = p_i(v_j) < e_{G_i,p_i}(v_j) = d_G(v_i,v_j) + p_i(v_i) \le 1 + rad(G) - \gamma$, we obtain $\gamma < 3$.
Hence, $v_i \in \layer_2$ and $v_\ell \in \layer_1$.
As $\ell < i$, necessarily $v_\ell$ is removed in iteration 2 step (c) of $\sigma$ construction.
However, this implies that $v_i$ is removed in iteration 2 step (d);
therefore, $i \geq y$, a contradiction that proves the claim.
\end{proof}

%\newpage
%\begin{algorithm}
%%\caption{ \\
%%Input: A distance-hereditary graph $G$ and the pendant/vertex ordering $\sigma$
%%as obtained from~\cite{DAMIAND200199}  \\
%%Output: All eccentricities of $G$
%%}
%\label{algo1}
%\begin{algorithmic}[1]
%\Procedure{compute-DHG-Eccentricities}{$G$}
  %\State Let $v_n$ be a central vertex of $G$
  %\State Let $\sigma$ be a pendant/vertex ordering produced by a breadth-first search tree rooted at $v_n$ (see~\cite{DAMIAND200199})
  %\State $p(v) \gets 0$ for all $v \in V$
  %\For{$i \gets 1, ..., z - 1$}
    %\If{$v_i$ is a pendant to $v_j$ in $G_i$}
      %\State $v_i.old \gets p(v_j)$
      %\State $p(v_j) \gets \max\{1 + p(v_i), p(v_j)\}$ 
    %\EndIf
    %\If{$v_i$ is a twin to $v_j$ in $G_i$ and $p(v_i) > p(v_j)$}
      %\State \Call{swap}{$v_i$,$v_j$}
    %\EndIf
  %\EndFor
  %\State \Call{compute-Tree-Weighted-Eccentricities}{$G_z, p$}
  %\For{$i \gets z - 1, ..., 1$}
    %\If{$v_i$ is a pendant to $v_j$ in $G_i$}
      %\State $p(v_j) \gets v_i.old$
      %\State $e(v_i) \gets \max\{p(v_i), 1 + e(v_j)\}$
    %\EndIf
    %\If{$v_i$ is a twin to $v_j$ in $G_i$ and $p(v_i) > p(v_j)$}
      %\State $e(v_j) \gets \max\{p(v_i) + d_G(v_i,v_j), e(v_j)\}$
      %\State $e(v_i) \gets \max\{p(v_j) + d_G(v_i,v_j), e(v_j)\}$
    %\EndIf
  %\EndFor
%\EndProcedure
%\end{algorithmic}
%\end{algorithm}
%\begin{theorem}
%Algorithm~\label{algo1} computes all eccentricities of a distance-hereditary graph $G$ in linear time.
%\end{theorem}
%\begin{proof}
%\end{proof}
%}%%%hnew

\section{Concluding remarks}
We have shown that the eccentricity function in distance-hereditary graphs is almost unimodal. We used this result to fully characterize centers of distance-hereditary graphs and to provide several bounds on the eccentricity of a vertex. Finally, a new linear time algorithm to calculate all eccentricities is presented.

\medskip\noindent
{\bf Acknowledgments.} We are very grateful to anonymous referees for many useful suggestions and to Guillaume Ducoffe for informing us about the relevant results from \cite{Coudert_2019}.  

%\listoftodos
\bibliographystyle{plain}
\bibliography{bibliography}

\end{document}